\documentclass[11pt]{article}

\usepackage{fullpage}
\usepackage{amsmath,amsfonts,amsthm,xspace,hyperref,graphicx,relsize,bm}
\usepackage{mathrsfs}
\usepackage{mathpazo}
\usepackage[margin=1in]{geometry}
\usepackage{float}
\usepackage{endnotes}
\usepackage{enumitem}
\usepackage{color}
\usepackage{amssymb,latexsym}
\usepackage{multicol}

\allowdisplaybreaks

\newtheorem{theorem}{Theorem}

\newtheorem{lemma}[theorem]{Lemma}

\newtheorem{fact}[theorem]{Fact}

\newtheorem{definition}[theorem]{Definition}
\theoremstyle{definition}

\newtheorem{remark}[theorem]{Remark}

\newcommand{\beq}{\begin{eqnarray}}
\newcommand{\eeq}{\end{eqnarray}}

\newcommand{\ket}[1]{{|#1\rangle}}
\newcommand{\bra}[1]{{\langle#1|}}

\newcommand{\ip}[2]{\langle #1 | #2 \rangle}
\newcommand{\ketbra}[2]{|#1\rangle\! \langle #2|}
\newcommand{\Tr}{\mbox{\rm Tr}}
\newcommand{\Id}{\mathbb{I}}

\newcommand{\hilb}{\mathcal{H}}

\DeclareMathOperator*{\Ex}{\mathbb{E}}

\newcommand{\states}{\mathrm{S}}
\newcommand{\dens}{\mathrm{D}}
\newcommand{\unitary}{\mathrm{U}}
\newcommand{\linmap}{\mathrm{L}}
\newcommand{\metalinmap}{\mathrm{T}}
\newcommand{\isometry}{\mathrm{J}}

\newcommand{\wt}[1]{\widetilde{#1}}
\newcommand{\what}[1]{\widehat{#1}}

\newcommand{\cM}{\ensuremath{\mathcal{M}}}

\newcommand{\C}{\ensuremath{\mathbb{C}}}

\newcommand{\R}{\ensuremath{\mathbb{R}}}
\newcommand{\Z}{\ensuremath{\mathbb{Z}}}

\DeclareMathOperator{\poly}{poly}

\newcommand{\eps}{\varepsilon}

\newcommand{\conj}[1]{\overline{#1}}

\newcommand{\advclass}{{\mathscr{A}}}

\newcommand{\flagspace}{\mathcal{F}}
\newcommand{\keyspace}{\mathcal{K}}
\newcommand{\msgspace}{\mathcal{M}}
\newcommand{\authspace}{\mathcal{Y}}
\newcommand{\auxspace}{\mathcal{Z}}
\newcommand{\tagspace}{\mathcal{T}}

\newcommand{\advspace}{\mathcal{Z}}
\newcommand{\superop}{{\mathcal{O}}}

\newcommand{\safespace}{\mathcal{S}}

\newcommand{\id}{{\mathcal{I}}}

\newcommand{\basis}{{\mathcal{B}}}

\newcommand{\auth}{{\mathsf{Auth}}}
\newcommand{\ver}{{\mathsf{Ver}}}

\newcommand{\QFT}{{\mathrm{QFT}}}

\newcommand{\rej}{\mathrm{REJ}}
\newcommand{\acc}{\mathrm{ACC}}

\newcommand{\attack}{\mathcal{E}}

\newcommand{\measure}{{\mathsf{Meas}}}

\begin{document}

\title{New security notions and feasibility results for authentication of quantum data}
\author{Sumegha Garg\thanks{\texttt{sumeghag@cs.princeton.edu}} \\Princeton \and Henry Yuen\thanks{\texttt{hyuen@cs.berkeley.edu}}\\UC Berkeley \and Mark Zhandry\thanks{\texttt{mzhandry@princeton.edu}}\\Princeton}
\date{}
\maketitle

\begin{abstract} We give a new class of security definitions for authentication in the quantum setting. These definitions capture and strengthen existing definitions of security against quantum adversaries for both \emph{classical} message authentication codes (MACs) and well as full quantum state authentication schemes. 
The main feature of our definitions is that they precisely \emph{characterize} the effective behavior of any adversary when the authentication protocol accepts, including correlations with the key. Our definitions readily yield a host of desirable properties and interesting consequences; for example, our security definition for full quantum state authentication implies that the entire secret key can be re-used if the authentication protocol succeeds.


Next, we present several protocols satisfying our security definitions. We show that the classical Wegman-Carter authentication scheme with $3$-universal hashing is secure against superposition attacks, as well as adversaries with quantum side information. We then present conceptually simple constructions of full quantum state authentication.

Finally, we prove a lifting theorem which shows that, as long as a protocol can securely authenticate the maximally entangled state, it can securely authenticate any state, even those that are entangled with the adversary. Thus, this shows that protocols satisfying a fairly weak form of authentication security automatically satisfy a stronger notion of security (in particular, the definition of Dupuis, et al (2012)).

%
%
%
%
\end{abstract}

\newpage
\tableofcontents

\newpage

\section{Introduction}
\label{sec:intro}

Authenticating messages is a fundamental operation in classical cryptography.  A sender Alice wishes to send a message $m$ over an insecure channel to a receiver Bob, with the guarantee that the message was not tampered with in transit. To accomplish this, Alice appends a ``signature'' $\sigma$ to $m$ using a shared secret key $k$ and send the message/signature pair $(m,\sigma)$ to Bob.  Bob receives some potentially altered pair $(m',\sigma')$, and then verifies that $\sigma'$ is a valid signature of $m'$ under key $k$.  If verification passes, Bob accepts $m'$, and if verification fails, Bob ignores the message and discards it.  A secure authentication protocol guarantees the following: even if the adversary has arbitrarily tampered with the communication channel, as long as the adversary does not know the secret key $k$, then either Bob rejects with high probability, or the message he receives is $m$. Such a (symmetric key) authentication protocol is usually referred to as a Message Authentication Code (MAC).  As long as $k$ is only used to authenticate a single message, information-theoretic security can be achieved: no adversary -- even a computationally unbounded one -- can modify the message without detection~\cite{wegman1981new}. 

Just as authentication is a fundamental operation in classical cryptography, it will continue to be an important tool in the coming age of quantum computers.  In this work, we investigate authentication in the quantum setting, and consider quantum attacks on both \emph{classical} authentication protocols, as well as full-fledged \emph{quantum} protocols for authenticating quantum data. What kinds of security guarantees can we hope for in the quantum setting? Various notions of security for authentication schemes against quantum attacks have been considered in the literature. However, as we will discuss below, these existing definitions do not fully capture security properties we would expect of a secure authentication scheme.

The contribution of our paper is three-fold: first, we present new security definitions for authentication in a quantum setting that strengthen previous definitions and address their limitations. Second, we prove interesting consequences of our stronger security definition for quantum authentication, such as information-theoretic key recycling and an easy protocol for quantum key distribution. Finally, we prove that several natural authentication protocols satisfy our security definitions.

\subsection{Quantum Attacks on Classical Protocols.}  A recent series of works~\cite{Boneh2011qrom,damgaard2013superposition,boneh2013quantum,boneh2013secure,zhandry2012qprfs,kaplan2016breaking} have studied quantum superposition attacks on classical cryptosystems.  In the setting of MACs, an adversary in such an attack is able to trick the sender into signing a superposition of messages.\footnote{One motivation for studying superposition attacks comes from the ``Frozen Smart-Card'' example~\cite{gagliardoni2015semantic}: real-world classical authentication systems are frequently implemented on small electronic devices such as RFID tags or a smart-cards. A determined and sophisticated attacker in possession of such a smart-card could try to perform a quantum ``side-channel attack'' on it: he places the device in a very low temperature environment, and attempts to query the device in quantum superposition. One would like to guarantee that even then the attacker is unable to, say, extract the secret key.}
 That is, the sender computes the map $\ket{m}\mapsto\ket{m,\sigma_m}$ in superposition, where $\sigma_m$ is the signature on $m$.  The adversary chooses some message superposition $\sum_m \alpha_m \ket{m}$, and the sender then applies the map, giving the adversary $\sum_m \alpha_m \ket{m,\sigma_m}$.  
At this point, it is unclear what the security definition should actually be.  Clearly, the adversary can tamper with the signed state: he can, for example, measure the entire state in the standard basis, obtaining the pair $(m,\sigma_m)$ with probability $|\alpha_m|^2$. Then $(m,\sigma_m)$ will pass verification, but will be different from the signed state the adversary received.  If the adversary can change the message state, what sort of guarantees can we hope for?

Boneh and Zhandry~\cite{boneh2013quantum} give the first definition of security for classical authentication against superposition attacks.   They argue that, at a minimum, the adversary given a single signed superposition should only be able to produce a single signed message; he should not be able to simultaneously produce two valid signed messages $(m,\sigma_m)$ and $(m',\sigma_{m'})$ for $m\neq m'$.  In the classical setting, this requirement is equivalent to the traditional MAC security definition.

However, the Boneh-Zhandry definition has some unsatisfying properties. For example, consider the case where the sender only signs messages that start with the email address of some intended recipient, say, \texttt{bob@gmail.com}. Suppose the adversary tricks the sender into a signing a superposition of messages that all begin with \texttt{bob@gmail.com}, but then manipulates the signed superposition into a different superposition that includes valid signed messages that \emph{do not} start with \texttt{bob@gmail.com}. Clearly, this is an undesirable outcome. Unfortunately, the Boneh-Zhandry definition does not rule out such attacks --- it only disallows an adversary from producing $q + 1$ valid signed messages when given $q$ signed superpositions. The situation illustrated here, however, is that the adversary is given \emph{one} signed superposition, and now wants to produce \emph{one} valid signed message that was not part of the original superposition.


Along similar lines, suppose an adversary tricks the sender into signing a uniform superposition on messages, and then produces a classical signed message $(m,\sigma)$.  From the sender's perspective, each message has weight $\frac{1}{|\cM|}$, where $\cM$ is the message space.  The sender cannot prevent the adversary from measuring the message state to produce $(m,\sigma)$ for a random $m$.  However, it is reasonable to insist as a security requirement that the adversary cannot bias the output of this measurement to obtain, say, $(m^*,\sigma_{m^*})$ with probability much higher than $\frac{1}{|\cM|}$.  Again, Boneh and Zhandry's definition does not preclude such a biasing, since the adversary only ever obtains a single signed message.  Thus, the Boneh-Zhandry definition does not capture natural non-malleability properties one would hope for from an authentication scheme.

Boneh and Zhandry's definition suffers from these weaknesses because it only considers what types of outputs the adversary can produce, ignoring the relationships between the output and the original signed state.  In the classical setting, the two approaches are actually equivalent, but in the quantum setting this is not the case.

\subsection{Quantum Authentication of Quantum Data.} We turn to the setting of schemes for authenticating quantum states. Barnum et al.~\cite{barnum2002authentication} was the first to study this, and they present a definition of non-interactive quantum authentication where, conditioned on the protocol succeeding, the sender has effectively teleported a quantum state to the receiver (provided that the probability of success is not too small).
%
%
%
They then give a scheme (called the \emph{purity testing scheme}) which attains this definition. Interestingly, they also show that quantum state authentication also implies quantum state \emph{encryption}.\footnote{By contrast, in the classical setting, message authentication does \emph{not} imply message encryption.} Subsequent works~\cite{ben2006secure,aharonov2008interactive,dupuis2012actively,broadbent2013quantum} presented some stronger security definitions that we will discuss momentarily. 

Roughly speaking, a (private-key) quantum authentication scheme is a pair of keyed quantum operations $(\auth_k,\ver_k)$, where $k$ is a secret key shared by the sender and receiver, where $\auth_k$ is a map that takes in a quantum message state $\rho$, and outputs an authenticated state $\sigma$. The map $\ver_k$ is a verification operation that takes in a (possibly) tampered state $\wt{\sigma}$ and outputs a state $\wt{\rho}$, along with a flag $\acc$ or $\rej$ indicating whether the verification succeeded or failed. These maps are such that for all input states $\rho$ and all keys $k$, we have $\ver_k(\auth_k(\rho)) = \rho \otimes \ketbra{\acc}{\acc}$. 

Barnum, et al. define an $\eps$-secure authentication scheme to be such that for all pure message states $\ket{\psi}$, for all adversary operations $\superop$, we have that 
\begin{equation}
	\Tr \Big \{ \Big ( (\Id - \ketbra{\psi}{\psi}) \otimes \ketbra{\acc}{\acc} \Big) \tau \Big \} \leq \eps
\label{eq:barnum}
\end{equation}
where $\tau = \frac{1}{|\keyspace|} \sum_k \ver_k(\superop(\auth_k(\ketbra{\psi}{\psi})))$ is the final state (including the $\acc$ or $\rej$ flag) after authentication, adversary attack, and verification, averaged over the choice of secret key $k$. Informally, it states that either the receiver obtains a state that is close to $\ket{\psi}$, or rejects with high probability. Barnum, et al. also extend this definition to the setting when the input state is not pure but a mixed state $\rho$; it similarly guarantees that the final output state, \emph{conditioned} on verifier success, will be close to the input state $\rho$ (provided that the success probability is not too small). 

However, this security definition of~\cite{barnum2002authentication} has some shortcomings.

\medskip
\paragraph{Adversaries with quantum side information.} It does not handle the case of when the adversary has some \emph{quantum side information} about the message state. Let $(\auth_k,\ver_k)$ be some $\eps$-secure authentication scheme (according to the Barnum, et al. definition). Let $\msgspace$ denote the sender's message register that is to be authenticated. Suppose the attacker has a qubit register $\advspace$ such that the joint state of the $\advspace \msgspace$ registers is
 $$
 		\frac{1}{\sqrt{2}} \ket{0}^\advspace \ket{\psi_0}^\msgspace + \frac{1}{\sqrt{2}} \ket{1}^\advspace \ket{\psi_1}^\msgspace
 $$ 
where $\ket{\psi_0}$ and $\ket{\psi_1}$ are some orthogonal states. The adversary is thus \emph{entangled} with the message state, and the qubit $\advspace$ is its side information about the message. Note that the marginal input state (i.e. if we ignore the adversary's qubit) is the probabilistic mixture $\rho = \frac{1}{2} \ketbra{\psi_0}{\psi_0} + \frac{1}{2} \ketbra{\psi_1}{\psi_1}$. Using the secret key $k$ shared with the receiver, the sender applies $\auth_k$ to $\msgspace$ and transmits the authenticated state across the channel. The adversary then performs the following attack: controlled on the $\advspace$ qubit, it forwards the authenticated message state untouched (if the $\advspace$ qubit is $\ket{0}$), and otherwise replaces the authenticated state with some fixed junk state $\ket{\xi}$ and sends that instead (if the $\advspace$ qubit is $\ket{1}$). When the receiver performs the verification operation, the junk state will fail to pass with high probability, and thus \emph{conditioned} on success (which occurs with probability at least $1/2$), the receiver will have a state close to $\ket{\psi_0}$. However, this final state is far from the original mixture $\rho$ -- violating the conclusion of the Barnum, et al. security condition. In other words, the adversary significantly tampered with the message, but the receiver still accepted with relatively large probability.

Being unable to handle adversaries with quantum side information prevents the security definition from being \emph{composable}. In many situations we would like to use authentication not as a stand-along task, but as a primitive in a larger protocol -- indeed, quantum authentication \emph{has} been used as a primitive in schemes for delegated quantum computation, e.g.,~\cite{aharonov2008interactive,broadbent2013quantum}. Here, the ``adversary'' (which may be other components of the protocol) may generate the inputs to the authentication scheme, and thus be entangled with the message that is supposed to be authenticated. If authentication scheme satisfied a composable security definition, then we may use the security of the authentication primitive in a black box manner to argue the proper functioning of the larger protocol.

Recently, several works~\cite{dupuis2012actively,broadbent2016efficient} have proposed composable security definitions for quantum authentication -- that is, they handle adversaries with quantum side information. However, their security definitions have the following drawback:

\medskip
\paragraph{Averaging over the key.} The security definitions of~\cite{barnum2002authentication,dupuis2012actively,broadbent2016efficient} averages over the secret key shared between the sender and receiver. Suppose Alice sends Bob the authenticated state $\sigma_k = \auth_k(\rho)$ using key $k$. Bob receives a (possibly tampered) state $\wt{\sigma}_k$, and proceeds to verify the authentication. Let $\tau_k$ denote Bob's state \emph{conditioned} on successful verification. Roughly speaking, the security definitions of~\cite{barnum2002authentication,dupuis2012actively,broadbent2016efficient} refers to the \emph{average} state $\Ex_k \tau_k$; in particular, it states that $\Ex_k \tau_k$ is close to the original state $\rho$. Even putting aside the example given above (where the average state may not be close to $\rho$), this statement does not, by itself, imply that $\tau_k$ is close to the original state $\rho$ \emph{with high probability} over $k$, which is a stronger property, and quantifies how much the adversary learns about the secret key $k$. In other words, the definition does not \emph{a priori} rule out the possibility that each $\tau_k$ is far from $\rho$ for every $k$, yet the average happens to be $\rho$.

Later, we will show how taking into account the correlations between the key and the final state of the protocol yields interesting consequences -- such as the ability to reuse the key upon successful verification. Furthermore, security definitions that average over the key are strictly weaker than those that keep the key register (we give a brief argument for this in Section~\ref{sec:framework}).


\subsection{This Work}

In this work, we address the above limitations by giving new security notions for authentication in the quantum setting. More generally, we present an abstract framework of security for both classical and quantum authentication schemes that not only captures existing security definitions (such as the Boneh-Zhandry definition for classical protocols or the Barnum, et al. definition of quantum state authentication), but also is more demanding in that it strongly \emph{characterizes} the (effective) behavior of an adversary. In particular, the adversary may have access to quantum side information with the message state that is being authenticated. The characterization of the adversary's admissable actions is what allows us to easily deduce many desirable security properties (such as unforgeability, key reuse, and more). Furthermore, we will show that various natural authentication protocols satisfy our security definitions.

Our abstract security framework is inspired by the simulation paradigm in classical cryptography.  In our framework, one first defines a class $\advclass$ of \emph{ideal adversaries}. Intuitively, ideal adversaries are those that cannot be avoided in a real execution of an ideal authentication protocol, such as those that discard messages, or ones that carry out actions explicitly allowed by the protocol. For example, in the case of classical protocols, one can define the class of ideal adversaries to be ones that ``behave classically'' on the message state -- that is, they're restricted to measurements in the computational basis. In the case of quantum authentication, an ideal adversary can \emph{only} act on the side information, but otherwise acts as the identity on the authenticated message. 

An authentication protocol $P$ satisfies our security definition with respect to the class $\advclass$ if the behavior of any adversary (not necessarily ideal) in the protocol $P$ can be approximately simulated by an ideal adversary in $\advclass$. We take the most general notion of simulation possible: the joint state of the secret key, the message state after the receiver's verification procedure, and the quantum side information held by the adversary \emph{conditioned on successful verification} must be indistinguishable from the same joint state arising from the actions of \emph{some} ideal adversary from the class $\advclass$. Since our notion of simulation is so general, this implies that our security definitions satisfy security under \emph{sequential composition}; that is, the authentication protocols that realize our security definition can be securely composed with arbitary cryptographic protocols in a sequential fashion.


We now discuss how security for both classical authentication schemes and fully quantum authentication protocols can be defined in this framework.

\subsubsection{A new security definition for classical authentication}  
The Boneh-Zhandry definition focuses on what classical signed messages an adversary can produce, treating the superposition access to the sender as a tool to mount stronger attacks.  Here, we instead think of a classical protocol giving rise to a weak form of authentication of quantum messages, where a superposition is authenticated by classically signing each message in the superposition.  That is, a state $\sum_m \alpha_m \ket{m}$ is authenticated as the state $\sum_m \alpha_m \ket{m,\sigma_m}$.  The state is similarly verified in superposition by running the classical verification algorithm in superposition.

More generally, we think of the protocol acting on message states that may be entangled with an adversary. For example, the sender could sign the $\msgspace$ part of the state $\sum_m \alpha_m \ket{m}^{\msgspace} \otimes \ket{\varphi_m}^{\advspace}$, where the adversary has control of the $\ket{\varphi_m}^{\advspace}$ states. The signed state then would become $\sum_m \alpha_m \ket{m, \sigma_m}^{\msgspace \tagspace} \otimes \ket{\varphi_m}^{\advspace}$. Signing mixed states can also be expressed in this way, simply by purifying the mixture. By thinking of the protocol in this way, we are able to give security definitions that actually consider the relationship between the sender's signed state and the final state the adversary produces.

Clearly, such a classical scheme cannot fully protect the quantum state.  An adversary could, for example measure $(m,\sigma_m)$, or any subset of bits of the state, and keep the result of such a measurement in his own private space. This would not be detected by the classical verification procedure, but the final message would have been changed. 

Our security definition for classical protocols says that, roughly, an arbitrary adversary can be simulated by an ideal adversary that can only do the following: perform some (partial) measurement of the message in the computational basis, and controlled on the outcome of the partial measurement, perform some quantum operation on his own private qubits. 
We also extend the definition to handle side information the adversary may have about the message state; for example, the adversary may possess the purification of the message state.  Thus, our definition is essentially the best one could hope for, since is disallows the adversary from doing anything other than operations that are trivially possible on \emph{any} classical protocol.

Our definition readily implies the Boneh-Zhandry security definition for one-time MACs, and does not suffer from the weakness of their definition\footnote{One limitation of our definition is that we consider the signature registers as being initialized by the signer.  Boneh and Zhandry, in contrast, allow the registers to be initialized by the adversary, with the signature being XORed into the registers}.  Finally, we show that the classical Wegman-Carter MAC that uses three-universal hashing is sufficient for achieving this strong security definition.

\subsubsection{Definitions for Quantum Authentication}  

We next turn to quantum protocols for authenticating quantum messages.  For general quantum protocols, the adversary can always do the following. He can always act non-trivially on his own private workspace -- the verification procedure can never detect this. Otherwise, he can forward the authenticated state as is, without recording any information about the state, or he can send junk to the receiver. Our strongest definition of security -- which we call \emph{total authentication} -- says that this is essentially all an adversary can do in a secure quantum authentication protocol. In other words, a real adversary in a total authentication protocol can be approximated by an ideal adversary that behaves trivially on the authenticated state. 

Our definition strengthens the definitions of~\cite{barnum2002authentication,dupuis2012actively,broadbent2016efficient}: not only do we we consider side information about the plaintext state, we also allow the receiver's view to include the authentication key as well as whatever information the adversary may learn about the key.  The ideal adversary must approximate the real adversary, even considering the entire key.  In contrast, existing definitions trace out the key --- either partially or entirely --- and therefore do not directly consider  \emph{arbitrary} information the adversary may learn about the key. Our security definition of total authentication thus rules out the possibility of the adversary learning significant information about the key. This fact has interesting consequences:
\begin{enumerate}
\item \textbf{Key reuse.} For example, our definition immediately implies that, upon successful verification by the receiver, the key can actually be completely recycled to authenticate a new message.  This is because, upon successful verification, the key is essentially independent of the adversary and can therefore be used again in the same protocol. This is contrast to the classical setting: in general keys cannot be recycled without computational assumptions. Furthermore, no prior definition for authentication of quantum data directly implies key re-usability, and no prior protocol for quantum messages gets full key re-usability upon successful verification. 

\item \textbf{A simple quantum key distribution protocol.} Our definition also gives a conceptually simple quantum key distribution (QKD) protocol\footnote{The observation that quantum authentication implies a form of QKD is due to Charlie Bennett.}.  Alice prepares a maximally entangled state, chooses a random key $k$, and authenticates half the state with the key.  She then sends the authenticated half to Bob, keeping the unauthenticated half to herself.  When Bob receives the state, he sends a ``received'' message back to Alice, who then sends the key $k$ to Bob.  Bob verifies the state using the key.  Even though the adversary eventually sees the authentication key $k$, he does not know the key when he intercepts the quantum state, and must therefore interact with the state without the key.  If Bob's verification passes, it implies, roughly, that the adversary could not have tampered with the state (by the security of total authentication); in particular, the adversary could not have learned any information about the maximally entangled state.  Therefore, Alice and Bob measure their halves of the maximally entangled state and obtain a shared key that is unknown to the eavesdropper.  If Bob's verification rejects, the two try again. Though this is not a practical QKD scheme (because any tampering by the adversary would cause Alice and Bob to abort), it is conceptually very simple and illustrates the power of our definitions. 
\end{enumerate}

\paragraph{A protocol satisfying total authentication.} We exhibit a protocol meeting our strong security notion. We present an authentication scheme based on \emph{unitary designs}, which are efficiently sampleable distributions over unitary matrices that behave much like the uniform distribution over unitaries when only considering low degree moments. 

\paragraph{Total authentication with key leakage.} We also give a definition of \emph{total authentication with key leakage}. This is a notion of security where the real adversary can be simulated by an ideal trivial adversary that only acts on its own private workspace, \emph{but in a manner that may depend on the key}. This is slightly weaker notion of security than total authentication, but it still implies simple QKD and some amount of key reuse. We note that the work of~\cite{hayden2011universal} essentially shows that the Barnum et al. protocol satisfies total authentication with (minor) key leakage. We also give a simple protocol that achieves this, based on the \emph{classical} Wegman-Carter authentication scheme.


\paragraph{A lifting theorem.} Finally, we prove an intriguing \emph{equivalence} between a very weak form of authentication security and a stronger notion. Specifically, this weak form of authentication security only guarantees that an authentication scheme is able to authenticate a \emph{single state}: a Bell state. Furthermore, this Bell state is unentangled with the adversary, and the security guarantee holds on average over the secret key. We prove a \emph{lifting theorem} that ``lifts'' this weak security to a much stronger one that shows the same authentication scheme, when augmented with a Pauli randomization step, is actually secure when authenticating \emph{arbitrary} messages, which might be entangled with the adversary! This stronger security notion still averages over the key, so it does not achieve total authentication. Nonetheless, we find it conceptually very interesting that such a lifting theorem holds.

%
%

\paragraph{Outline.} In the next section we cover some preliminaries and notation. In Section~\ref{sec:framework} we formally present the fundamental security definitions used in our paper. In Sections~\ref{sec:classical_properties} and~\ref{sec:quantum_properties} we prove that our definitions satisfy the properties expected of authentication schemes. In Section~\ref{sec:threewise}, we analyze the security of the Wegman-Carter MAC with 3-universal hashing within our security framework. In Section~\ref{sec:qft} we present and analyze the Auth-QFT-Auth scheme. In Section~\ref{sec:design} we present and analyze the unitary design scheme. In Section~\ref{sec:lifting} we prove the lifting theorem.

\section{Preliminaries}
\label{sec:prelim}

\subsection{Notation}
\paragraph{Quantum information.} We assume basic familiarity with quantum computing concepts, such as states, measurements, and unitary operations. We will use caligraphic letters to denote Hilbert spaces, such as $\hilb$, $\msgspace$, $\tagspace$, $\keyspace$, and so on. We write $\states(\hilb)$ to denote the set of unit vectors in $\hilb$. For two Hilbert spaces $\hilb$ and $\msgspace$, we write $\linmap(\hilb,\msgspace)$ to denote the set of matrices that map $\hilb$ to $\msgspace$. We abbreviate $\linmap(\hilb,\hilb)$ as simply $\linmap(\hilb)$. The following are important subsets of $L(\hilb)$ that we'll use throughout this paper. 
\begin{itemize}
	\item $\dens(\hilb)$ denotes the set of \emph{density matrices} on $\hilb$; that is, positive semidefinite operators on $\hilb$ with unit trace. 
	\item $\dens_{\leq}(\hilb)$ denotes the set of \emph{subnormalized} density matrices on $\hilb$; that is, positive semidefinite operators on $\hilb$ with trace at most one.
	\item $\unitary(\hilb)$ denotes the set of unitary matrices acting on $\hilb$. For an integer $N$, we will also write $\unitary(N)$ to denote the set of all $N \times N$ complex unitary matrices.
\end{itemize}
Another important class of operators are \emph{isometries}: these are like unitaries, except that can append ancilla qubits. We say that a map $V \in \linmap(\hilb,\msgspace)$ is an isometry if for all vectors $\ket{\psi} \in \hilb$, $\| V\ket{\psi} \| = \| \ket{\psi} \|$. Note that this requires $\dim(\msgspace) \geq \dim(\hilb)$. We will let $\isometry(\hilb,\msgspace)$ denote the set of isometries in $\linmap(\hilb,\msgspace)$.

We use $\Id$ to denote the identity matrix. For a Hilbert space $\hilb$, we let $|\hilb|$ denote the dimension of $\hilb$. 

We will typically decorate states and unitaries with superscripts to denote which spaces they act on. For example, let $\authspace$ and $\advspace$ be two Hilbert spaces. Let $U \in \unitary(\authspace)$ and let $V \in \unitary(\authspace \otimes \advspace)$. Then when we write the product $U^{\authspace} V^{\authspace \advspace}$ we mean the $(U^{\authspace} \otimes \Id^{\advspace}) V^{\authspace \advspace}$; we will often omit mention of the identity unitary when it is clear from context.

Another convention is the implicit partial trace. For example, let $\rho^{\keyspace \msgspace} \in \dens(\keyspace \otimes \msgspace)$. Then $\rho^{\msgspace} = \Tr_{\keyspace}(\rho^{\keyspace \msgspace})$. Additionally, given a pure state $\ket{\rho}$, we will let $\rho$ denote the rank one density matrix $\ketbra{\rho}{\rho}$. 

\paragraph{Superoperators.} In this paper we will consider \emph{superoperators}, which are linear maps that act on a vector space of linear maps. For Hilbert spaces $\hilb$ and $\msgspace$, let $\metalinmap(\hilb,\msgspace)$ denote the set of all linear maps that take elements of $\linmap(\hilb)$ to $\linmap(\msgspace)$. While superoperators can be very general, we will focus on superoperators $\superop \in \metalinmap(\hilb,\msgspace)$ that are \emph{completely positive} and \emph{trace non-increasing}, which have the following characterization: there exists an alphabet $\Sigma$ and set of matrices (not necessarily Hermitian) $\{ A_a \}_{a \in \Sigma} \subset \linmap(\hilb,\msgspace)$ such that
\begin{enumerate}
	\item $\superop(X) = \sum_{a \in \Sigma} A_a X A_a^\dagger$	for all $X \in \linmap(\hilb)$, and
	\item $\sum_{a \in \Sigma} A_a^\dagger A_a \preceq \Id^\hilb$.
\end{enumerate}
For the rest of this paper, when we speak of superoperators, we will always mean completely positive, trace non-increasing superoperators. Although the definition of superoperators is rather abstract, they capture general quantum operations on arbitrary quantum states, including post-selection, as demonstrated by Stinespring's dilation theorem:
\begin{theorem}[Stinespring's dilation theorem]
\label{thm:stinespring}
A map $\superop \in \metalinmap(\hilb,\msgspace)$ is a completely positive, trace non-increasing superoperator if and only if there exists auxiliary Hilbert spaces $\advspace,\advspace'$, an isometry $V \in \isometry(\hilb \otimes \advspace,\msgspace \otimes \advspace')$, and a projector $\Pi$ acting on $\msgspace \otimes \advspace'$ such that for all density matrices $\rho \in \dens(\hilb)$, we have
$$
	\superop(\rho) = \Tr_{\advspace'}( \Pi V \rho V^\dagger \Pi).
$$
\end{theorem}

%
%

\paragraph{Matrix norms and distance measures.} We will make use of several matrix norms and distance measures in this paper. 

Given a (not necessarily unit) vector $\ket{\psi} \in \hilb$, we use $\| \ket{\psi} \|_2$ to denote the Euclidean norm of $\ket{\psi}$. 

The most matrix norm important is the \emph{trace norm} of a linear operator $X$, defined to be $\| X \|_1 = \Tr(\sqrt{X^\dagger X})$. Correspondingly, the \emph{trace distance} between density matrices $\rho,\sigma$ is defined to be $\| \rho - \sigma \|_1$. The operational significance of the trace distance is that $\| \rho -\sigma \|_1$ denotes the maximum bias with which one can distinguish between $\rho$ and $\sigma$ using any quantum operation.

The next norm we will make use of is the \emph{Frobenius norm} of a linear operator $X$, which is defined to be $\|X\|_2 = \sqrt{\Tr(X^\dagger X)}$. A useful property of the Frobenius norm is that $\|X\|_2 = \sqrt{\sum_{ij} |X_{ij}|^2}$, where the sum is over all the matrix entries of $X$ (with respect to any basis).

The \emph{operator norm} (also known as the \emph{spectral norm}) of an operator $X \in \linmap(\hilb)$ is to defined to be $\|X \|_\infty = \sup_{\ket{v} \in \states(\hilb)} \| X\ket{v} \|_2$, where the supremum is over all unit vectors in $\hilb$.

\begin{fact}
\label{fact:pure_to_dens}
	Let $\ket{\psi},\ket{\theta} \in \states(\hilb)$. Then 
	$$
		\| \psi - \theta \|_1 \leq 2 \| \ket{\psi} - \ket{\theta} \|_2
	$$
	where recall that $\psi = \ketbra{\psi}{\psi}$ and $\theta = \ketbra{\theta}{\theta}$.
\end{fact}

\subsection{Basic definitions for authentication}
\label{sec:defs}

\paragraph{Spaces.} We let $\keyspace$ denote the \textbf{key space}, $\msgspace$ denote the \textbf{message space}, $\authspace$ denote the \textbf{authenticated space}, and $\flagspace$ to denote the \textbf{flag space}. The flag space $\flagspace$ is a two-dimensional Hilbert space spanned by orthogonal states $\ket{\acc}$ and $\ket{\rej}$. The space $\advspace$ is the \textbf{private space of the adversary}. We will let $\safespace$ denote the a registers held by the sender and receiver that, during the execution of the authentication protocol, are not communicated nor acted upon by the sender, receiver, or adversary.

\paragraph{Authentication scheme.}
An authentication scheme is a pair of keyed superoperators $\auth,\ver$ where
\begin{itemize}
	\item $\auth_k$ for $k\in\keyspace$ is a superoperator mapping $\dens(\msgspace)$ to $\dens(\authspace)$.
	\item $\ver_k$ for $k\in\keyspace$ is a superoperator mapping $\dens(\authspace)$ to $\dens(\msgspace \otimes \flagspace)$.
\end{itemize}
satisfying the correctness requirements that for any quantum state $\rho\in\dens(\msgspace)$, for all keys $k \in \keyspace$, $\ver_k(\auth_k(\rho)) = \rho \otimes \ketbra{\acc}{\acc}$.

We will also use $\auth$ and $\ver$ to denote the operators
$$
	\auth(\cdot) = \sum_k \ketbra{k}{k} \otimes \auth_k(\cdot) \qquad \qquad \ver(\cdot) = \sum_k \ketbra{k}{k} \otimes \ver_k(\cdot).
$$

\paragraph{Some simplifying assumptions.} This definition of authentication scheme is more general than we need in this paper. Throughout this work, we shall work with a simplified model of authentication schemes: first, we will assume that $\auth_k$ behaves as an isometry taking $\msgspace$ to $\authspace$ (i.e. it isn't probabilistic). Let $\mathcal{V}_k$ denote the subspace of the Hilbert space $\authspace$ that is the image of $\auth_k$, let $\Pi_{\mathcal{V}_k}$ denote the projector onto the space $\mathcal{V}_k$, and let $\auth_k^{-1}$ denote the inverse isometry that maps $\mathcal{V}_k$ to $\msgspace$. In this case, a canonical way to define the $\ver_k$ superoperator is as follows:
\begin{align}
\label{eq:ver_full}
	\rho \mapsto \big ( \auth_k^{-1} \circ \Pi_{\mathcal{V}_k} \big)(\rho) \otimes \ketbra{\acc}{\acc}^{\flagspace} + \Tr((\Id - \Pi_{\mathcal{V}_k}) \, \rho) \, \frac{\Id_{\msgspace}}{|\msgspace|} \otimes \ketbra{\rej}{\rej}^{\flagspace}
\end{align}

Here, $\auth_k^{-1} \circ \Pi_{\mathcal{V}_k}$ denotes the operation that first applies the projection $\Pi_{\mathcal{V}_k}$ to the state, followed by the inverse isometry $\auth_k^{-1}$. The state $\frac{\Id_{\msgspace}}{|\msgspace|}$ is the maximally mixed state on the message space. In other words, the verification procedure first checks that the received state (which resides in $\authspace$) is supported on the subspace of valid signed states $\mathcal{V}_k$. If so, then it inverts the authentication isometry to obtain an unsigned message state, and sets the $\flagspace$ register to $\ket{\acc}$. Otherwise, it replaces the state with a uniformly random message state, and sets the $\flagspace$ register to $\ket{\rej}$. 

However in this paper we are mostly concerned with the output of the $\ver_k$ procedure in the \emph{accepting} case. For technical convenience then, throughout this paper we will treat $\ver_k$ as the following superoperator mapping $\dens(\authspace)$ to $\dens(\msgspace)$:
$$
	\ver_k(\rho) = \big ( \auth_k^{-1} \circ \Pi_{\mathcal{V}_k} \big)(\rho).
$$
In other words, it only outputs the $\ket{\acc}$ part of~\eqref{eq:ver_full}, and does not output a $\acc$ or $\rej$ flag. Furthermore, notice that this superoperator is not trace preserving; the trace of $\ver_k(\rho)$ is equal to the probability that $\rho$ was accepted by the verification procedure defined in~\eqref{eq:ver_full}. Thus one can view $\ver_k$ as a ``filter'' that only accepts states that were properly authenticated. 

We stress, however, that these simplifying assumptions are not crucial to our results -- it is mostly for notational convenience that we treat $\ver_k$ as a filter.

\paragraph{Classical Authentication.} In a classical authentication protocol, the authentication operator $\auth_k$ is specified by a classical function $\auth_k:\msgspace \mapsto \authspace$  acting on the computational basis, run in superposition on the input state. The verification operator behaves the same as described above: $\ver_k$ projects onto the subspace of $\authspace$ spanned by classical strings $\auth_k(m)$ for all $m \in \msgspace$, and then applies the inverse map $\auth_k^{-1}$.

\paragraph{Message authentication codes.} A message authentication code (or MAC) is special type of classical authentication scheme $(\auth,\ver)$ where for a message $m$, $\auth_k(m) = (m,\sigma(k,m))$, where we call $\sigma(k,m)$ the \emph{message tag}.  We treat $\ver_k$ as an operator that projects out messages that do not have valid tags, and for messages with valid tags, $\ver_k$ will strip the tags away:
$$
	\ver_k = \sum_m \ketbra{m}{m,\sigma(k,m)}.
$$

\paragraph{Adversaries.} We model adversaries in the following way: the adversary prepares the initial message state $\ket{\rho}^{\msgspace \safespace \advspace}$, where we can assume that the adversary possesses the purification of $\rho^{\msgspace \safespace}$. After the state is authenticated with a secret key $k$, the adversary gets to attack the $\authspace \advspace$ spaces with an arbitrary completely positive trace non-increasing superoperator $\superop$. After this attack, the state is un-authenticated with the same key $k$.


We don't require the superoperator $\superop$ to be trace preserving; this is to allow adversaries to \emph{discard} certain measurement outcomes (or, alternatively, \emph{post-select} on measurement outcomes, without renormalizing). While this may seem to give the adversary far too much power, in our security definitions we take into account the probability of the event that the adversary post-selects on. If this probability is too small, the security guarantees are meaningless, which is necessary. Allowing for superoperators to be trace non-preserving will help make our definitions clean to state.

\paragraph{A remark about the sender and receiver's private register $\safespace$.} The reader may wonder why we do not allow the sender, receiver, nor adversary to act upon the $\safespace$ register during the execution of the authentication protocol. The register $\safespace$ is supposed to model entanglement the sender and receiver may keep during the protocol. The important aspect of it is that the adversary does \emph{not} have access to this side information.

If, when analyzing the authentication scheme in the context of a larger protocol in which the sender/receiver \emph{do} act upon the register $\safespace$, we can assume that during the authentication phase, the sender and receiver do not touch $\safespace$, but wait until the authentication protocol is over. Thus we can analyze the behavior of the authentication protocol without this action. 

\section{Security Framework for Quantum Authentication}
\label{sec:framework}

We present our security definitions using the real/ideal paradigm. Let $(\auth,\ver)$ be an authentication protocol, with key space $\keyspace$, message space $\msgspace$, and authenticated space $\authspace$. 

\begin{definition} 
\label{def:general_sec}
Let $(\auth,\ver)$ be an authentication scheme. Let $\mathscr{A} \subseteq \metalinmap(\authspace \advspace,\authspace \advspace)$ denote a set of ideal adversaries. The scheme $(\auth,\ver)$ is \emph{$\eps$-reduces to $\mathscr{A}$-adversaries} iff the following holds: for all initial message states $\ket{\rho}^{\msgspace \safespace \advspace}$, for all adversaries $\superop \in \metalinmap(\authspace \advspace,\authspace \advspace)$, there exists an ideal adversary $\id \in \mathscr{A}$ such that the following (not necessarily normalized) states are $\eps$-close in trace distance:
\begin{itemize}
	\item (Real experiment) $\Ex_k\ketbra{k}{k}\otimes \left[ \ver_k \circ \superop \circ \auth_k \right](\rho^{\msgspace \safespace \auxspace})$
	\item (Ideal experiment) $\Ex_k\ketbra{k}{k}\otimes \left[ \ver_k \circ \id \circ \auth_k \right] (\rho^{\msgspace \safespace \auxspace})$
\end{itemize}
where $\auth_k$ acts on $\msgspace$, $\ver_k$ acts on $\authspace$, and both act as the identity on $\safespace \advspace$.
\end{definition}

Intuitively, our security definition states that for an authentication scheme $(\auth,\ver)$ that is $\mathscr{A}$-secure, for all initial message states $\rho^{\msgspace \safespace \advspace}$, an \emph{arbitrary} adversary that acts on an authenticated state $\auth_k(\rho^{\msgspace  \safespace  \advspace})$ is \emph{reduced} to an ``ideal adversary'' in $\advclass$; behaving differently will cause the verification procedure to abort. In other words, ``all the adversary can do'' is behave like some adversary in the class $\advclass$.


\paragraph{A comment about normalization.} It is important that the states of the real experiment and ideal experiment are not requiried to have unit trace. This is because their trace corresponds to the probability that the verification procedure accepts. If the probability of acceptance is smaller than $\eps$, then the security guarantee is vacuous. Intuitively, this corresponds to situations such as the adversary successfully guessing the secret key $k$, so we cannot expect any security guarantee in that setting. However, if the probability of acceptance is significantly larger than $\eps$, then we can condition on acceptance, and still obtain a meaningful security guarantee: the distance between the (renormalized) real experiment and ideal experiments is small.

\medskip \medskip

We now specialize the above definition to some important classes of ideal adversaries that we will consider in this paper. Note that for two classes of ideal adversaries $\advclass$ and $\advclass'$, if $\advclass \subset \advclass'$, then an authentication scheme reducing to $\advclass$-adversaries implies reducing to $\advclass'$-adversaries. Hence reducing to $\advclass$-adversaries is a stronger security guarantee.


\subsection{Basis-dependent authentication}
We first define a notion of security of authentication schemes that reduce to a \emph{basis-respecting} adversary. 

\begin{definition}[Basis-respecting adversaries] 
Let $\basis = \{ \ket{\psi} \}$ denote an orthonormal basis for $\authspace$. Then an adversary $\id \in \metalinmap(\authspace \advspace,\authspace \advspace)$ is \emph{$\basis$-respecting} iff it can be written as
$$
	\id(\sigma) = \Tr_{\advspace'}(\Pi V \sigma V^\dagger \Pi)
$$
for all $\sigma \in \dens(\authspace \advspace)$, where $\Pi$ is a projector acting on $\advspace \advspace'$, and $V \in \isometry(\authspace \advspace,\authspace \advspace \advspace')$ is an isometry that can be written as 
$$
	V = \sum_{\psi \in \basis} \ketbra{\psi}{\psi}^\authspace \otimes V_\psi
$$
where for each $\psi$, $V_\psi \in \isometry(\advspace,\advspace \advspace')$ is some isometry. 
\end{definition}
Without the second condition on $V$, by Stinespring's Dilation Theorem every superoperator can be written as $\id(\sigma) = \Tr_{\advspace'}(\Pi V \sigma V^\dagger \Pi)$ for some choice of isometry $V$ and projector $\Pi$. However, the second condition forces $V$ to respect the basis $\basis$. Intuitively, a basis-respecting adversary first performs some (partial) measurement on the $\authspace$ register in the $\basis$ basis, and based on the measurement outcome, performs some further isometry on the side information in $\advspace$. When $\basis$ is simply the computational basis, then the adversary treats the $\authspace$ register as classical. 

\begin{definition}[Security relative to a basis]
	Let $\basis$ be a basis for $\authspace$. An authentication scheme $(\auth,\ver)$ \emph{$\eps$-authenticates relative to basis $\basis$} iff it $\eps$-reduces to the class of $\basis$-respecting adversaries.
\end{definition}

Intuitively, our new definition captures the ``best possible'' security definition for \emph{classical} authentication protocols.  With a classical protocol, the adversary can perform arbitrary measurements on the authenticated space without detection by the verification algorithm.  Because measurements are now undetectable, the adversary can also perform $\sigma$-dependent operations to the auxiliary registers, where $\sigma$ is the classical authenticated message observed in the authenticated registers.  For example, he can copy $\sigma$ into the auxiliary space.  He can also now choose to abort or not depending on $\sigma$.  However, he should not be able to turn $\sigma$ into $\sigma'\neq \sigma$.

In Section~\ref{sec:classical_properties}, we establish consequences of our definition of basis-dependent security, including the property of unforgeability: the adversary cannot produce two valid signed messages with non-negligible probability, when given access to only one superposition. Thus, our definition subsumes the Boneh-Zhandry security definition for one-time MACs.



In Section~\ref{sec:threewise} we show that the classical Wegman-Carter MAC where the message $m$ is appended with $h(m)$, where $h(\cdot)$ is drawn from a three-wise independent hash family, is a scheme that authenticates relative to the computational basis. 

\begin{theorem}
	The Wegman-Carter MAC with three-universal hashing is $O(\sqrt{|\msgspace|/|\tagspace|})$-authenticating relative to the computational basis, where $\tagspace$ is the range of the hash family. 
\end{theorem}

\subsection{Total authentication}

In this section we formally define our notion of total authentication. First, we define \emph{oblivious adversaries}. 

\begin{definition}[Oblivious adversary]
	An adversary $\id \in \metalinmap(\authspace \advspace,\authspace \advspace)$ is \emph{oblivious} iff there exists a superoperator $\superop \in \metalinmap(\advspace,\advspace)$ such that
	$$
		\id(\sigma) = (\Id^\authspace \otimes \superop)(\sigma)
	$$
	for all $\sigma \in \dens(\authspace \advspace)$.
\end{definition}

In other words, an oblivious adversary does not act at all on the authenticated message, and only acts on the auxiliary side information that it possesses about the state. 

\begin{definition}[Total authentication]
An authentication scheme $(\auth,\ver)$ \emph{$\eps$-totally authenticates} iff it $\eps$-reduces to the class of oblivious adversaries.	
\end{definition}

This generalizes the security definition of~\cite{dupuis2012actively}, which is similar, except it traces out the key register. Therefore, it does not keep track of potential correlations between the adversary and the key. We will argue shortly that our definition of total authentication is strictly stronger than the definition of~\cite{dupuis2012actively}; that is, there are protocols which satisfy the security definition of~\cite{dupuis2012actively}, but do not satisfy total authentication.


In Section~\ref{sec:quantum_properties} we establish a few properties of this definition. The first is that a totally authenticating scheme yields encryption of the quantum state. Barnum, et al. showed that quantum state authentication implies quantum state encryption~\cite{barnum2002authentication}. However, they did not take into account quantum side information. We show that our definition very easily implies encryption even when the adversary may be entangled with the message state. 

Then, we show how our notion of total authentication gives rise to a conceptually simple version of quantum key distribution (QKD).~\cite{hayden2011universal} have already observed that the universal composability of the Barnum et al. protocol implies that it can be used to perform QKD as well. Thus while our application of quantum authentication to QKD is not novel, we use this as another opportunity to showcase the strength of our definition. We also show how our definition easily implies full key reuse.

In Section~\ref{sec:design} we present a scheme achieves total authentication, and to our knowledge this is the first scheme that achieves such security. 

\begin{theorem}
	The unitary design scheme is $2^{-s/2}$-totally authenticating, where $s$ is the number of extra $\ket{0}$ qubits.
\end{theorem}

As a consequence, this yields an authentication scheme where the key can be recycled fully, conditioned on successful verification by the receiver. In contrast, the protocol of Barnum et al. is not known to possess this property;~\cite{hayden2011universal} showed that most of the key can be securely recycled. 

\subsection{Total authentication with key leakage}

\newcommand{\leak}{\ell}

Finally, we introduce a slight weakening of the definition of total authentication above: we consider schemes that achieve total authentication of quantum data, but incur some \emph{key leakage}. We model this in the following way: let $\keyspace'$ be such that $|\keyspace|' \leq |\keyspace|$. Define a \emph{key leakage function} $\ell: \keyspace \mapsto \keyspace'$. If $|\keyspace'|$ is strictly smaller than $|\keyspace|$, then $\ell(k)$ must necessarily lose information about the key $k \in \keyspace$, but it will also leak some information about it. 

In a total authentication scheme with key leakage, an arbitrary adversary is reduced to an oblivious adversary (i.e., is forced to only act on the side information), but the manner in which it acts on the side information \emph{may depend on $\ell(k)$}.

\begin{definition}[Authentication with key leakage]
Let $(\auth,\ver)$ be an authentication scheme.  Let $\keyspace'$ be some domain such that $|\keyspace'| \leq |\keyspace|$ and let $\leak:\keyspace \to \keyspace'$ be a key leakage function.  Let $\mathscr{A} \subseteq \metalinmap(\authspace \advspace,\authspace \advspace)$ denote a set of ideal adversaries. The scheme $(\auth,\ver)$ \emph{$\eps$-reduces to $\mathscr{A}$-adversaries with key leakage $\leak$} iff the following holds: for all initial message states $\ket{\rho}^{\msgspace \safespace \advspace}$, for all adversaries $\superop \in \metalinmap(\authspace \advspace,\authspace \advspace)$, there exists a collection of ideal adversaries $\{\id_h \} \subset \mathscr{A}$, indexed by $h \in \keyspace'$, such that the following (not necessarily normalized) states are $\eps$-close in trace distance:
\begin{itemize}
	\item (Real experiment) $\Ex_k\ketbra{k}{k}\otimes \left[ \ver_k \circ \superop \circ \auth_k \right](\rho^{\msgspace \safespace\auxspace})$
	\item (Ideal experiment) $\Ex_k\ketbra{k}{k}\otimes \left[ \ver_k \circ \id_{\ell(k)} \circ \auth_k \right](\rho^{\msgspace \safespace \auxspace})$.
\end{itemize}
\end{definition}

\begin{definition}[Total authentication with key leakage]
Let $\keyspace'$ be some domain such that $|\keyspace'| \leq |\keyspace|$ and let $\leak:\keyspace \to \keyspace'$ be a key leakage function. An authentication scheme $(\auth,\ver)$ \emph{$\eps$-totally authenticates with key leakeage $\leak$} iff it $\eps$-reduces to the class of oblivious adversaries with key leakeage $\leak$.	
\end{definition}

This definition may seem somewhat strange: how is an ideal adversary able to learn bits $\ell(k)$ of the key $k$, if it doesn't act on the authenticated part of the state at all? Of course, any adversary that learns something about the key must have acted on the authenticated state, but the point is that, conditioned on successful verification, the adversary ``effectively'' behaved like an oblivious adversary that had access to $\ell(k)$. 

In Section~\ref{sec:qft} we present a very simple scheme that achieves total authentication with some key leakage: to authenticate an arbitrary quantum state $\rho$, first apply the classical Wegman-Carter authentication scheme on it using key $k$. Then, apply $H^{\otimes n}$ to all the qubits in the authenticated state (i.e. apply the quantum Fourier transform over $\Z_2$). Finally, apply the classical Wegman-Carter scheme again using a fresh key $h$. Thus, we are authenticating the state $\rho$ in complementary bases. We call this the ``Auth-QFT-Auth'' scheme.

We will show that this in fact achieves total authentication (and hence encryption of the state), but at the cost of leaking the ``outer key'' $h$:

\begin{theorem}[Security of the Auth-QFT-Auth scheme]
The Auth-QFT-Auth scheme is $\delta$-totally authenticating with outer key leakage, where $\delta = O(\sqrt{|\msgspace|^{5/2}/|\authspace|})$. 
\end{theorem}

While this scheme leaks some bits of the outer key, it preserves the secrecy of the state $\rho$ and the ``inner key'' $k$. Furthermore, it is much more ``lightweight'' than the full unitary design scheme that achieves total authentication without key leakage. It also illustrates that applying a simple classical authentication scheme in complementary bases is already enough to reduce a full quantum adversary to performing only trivial attacks. Finally, the analysis of this scheme crucially relies on the basis dependent security definition above. 


\subsection{A remark about efficiency}

Recently, Broadbent and Wainewright~\cite{broadbent2016efficient} study the \emph{efficiency} of the simulating ideal adversaries in the security proofs of two authentication schemes, the Clifford scheme and the trap code scheme. Specifically, they show that if the adversary in the authentication protocol is a quantum computer that runs in time $T$, then the ideal adversary which simulates it also runs in time $O(T)$. This efficiency-preservation is important for notions of composable security.

We note that the constructions of the ideal adversary in our analysis of the Wegman-Carter scheme, the Auth-QFT-Auth scheme, and the unitary design scheme are also efficiency preserving, and hence if the arbitrary adversary runs in polynomial time, then the simulating adversary also runs in polynomial time.

\subsection{Comparison with security definition in~\cite{dupuis2012actively}}
Similarly to our definition, the security definition of message authentication~\cite{dupuis2012actively} implies that essentially all the adversary can do is act on its own private workspace. However, it traces out the key register, and thus it does not keep track of correlations between the adversary and the secret key. It is a natural question to ask whether the security definition of~\cite{dupuis2012actively} \emph{implies} our definition of total authentication. Here we show that it cannot, because there are protocols that satisfy the~\cite{dupuis2012actively} definition, but not ours.\footnote{See Section~\ref{sec:lifting} for a formal statement of the~\cite{dupuis2012actively} definition.}

Consider a protocol $(\auth,\ver)$ that satisfies the~\cite{dupuis2012actively} definition. Let $k$ denote the secret key used in the protocol. Now consider the following modified protocol $(\auth',\ver')$: to authenticate a message state $\rho$, it produces $\auth_k(\rho)$, but then appends an independently random bit $b$, where $(k,b)$ is the secret key register of $(\auth',\ver')$. To verify, the receiver just applies the $\ver_k$ operation, and ignores the last bit. This new protocol still satisfies the~\cite{dupuis2012actively} definition, because the extra bit $b$ is independent of the $(\auth,\ver)$ process, and thus final state of the protocol can be simulated by an ideal adversary that generates its own $b$ bit -- as long as we're tracing out the key. However, this protocol does not satisfy total authentication. This is because an adversary can simply copy the the bit $b$ into its private workspace; but this cannot be simulated by an ideal adversary that is unentangled with the $(k,b)$ register.


Furthermore, any authentication scheme satisfying ~\cite{dupuis2012actively}'s security definition also satisfies ''total authentication with key leakage" for some key leakage function $\leak$ and any authentication scheme satisfying "total authentication with key leakage" satisfies the key-averaged security definition. Hence, these two security definitions are equivalent (up to some error).

\section{Properties of basis-dependent authentication}
\label{sec:classical_properties}

\subsection{Indistinguishability from measured}

Here, we show that any classical scheme that authenticates relative to the computational basis implies that the authenticated state is indistinguishable from being measured in the computational basis. For concreteness we will work with the computational basis; this is without loss of generality.


\begin{theorem}\label{thm:indistfrommeasured} Let $\eps < 1/2$. If $(\auth_k,\ver_k)$ $\eps$-authenticates relative to the computational basis, then the following two states are $7\sqrt{\eps}$-close:
\begin{itemize}
		\item $\Ex_k \auth_k (\rho^{\msgspace \safespace \auxspace})$ (the unmeasured authenticated state), and 
		\item $\Ex_k \left[ \measure \circ \auth_k \right](\rho^{\msgspace \safespace \auxspace})$ (the measured authenticated state),
\end{itemize}
where $\measure$ denotes measuring $\authspace$ in the computational basis, and $\auth_k$ and $\measure$ both act as the identity on $\safespace \auxspace$.
\end{theorem}
\begin{proof}  We prove this theorem by contradiction: assuming an adversary can distinguish from measured, we will obtain a violation of the security of authentication.  
Let $\rho^{\msgspace \safespace \auxspace}$ be the input quantum state to the protocol. Let $D$ be a distinguisher violating the indistinguishability from measured property.  Suppose $D$ has very large distinguishing advantage $1-\gamma$ (we deal with the low-distinguishing advantage later). This means that 
\begin{itemize}
	\item $D(\Ex_k \auth_k (\rho^{\msgspace \safespace \auxspace}))$ outputs 1 with probability at least $1-\gamma$, and 
	\item $D(\Ex_k \left[\measure\circ \auth_k\right](\rho^{\msgspace \safespace \auxspace}))$ outputs 1 with probability at most $\gamma$
\end{itemize}
	
Now suppose we set up the following state, where $\safespace'$ is an extra qubit register:
$$
	\frac{1}{2} \left [ \ketbra{0}{0}^{\safespace'} \otimes \measure(\rho^{\msgspace \safespace \advspace}) + \ketbra{1}{1}^{\safespace'} \otimes \rho^{\msgspace \safespace \advspace} \right].
$$
In other words, the $\safespace'$ qubit indicates whether we have measured $\msgspace$ in the computational basis or not. Next we authenticate this state using $\auth_k$. Since the scheme is classical, authentication commutes with measurement in the computational basis.  Therefore, the authenticated state, averaged over the key $k$, is 
\begin{equation}
\label{eq:classical_prop_1}
	\frac{1}{2} \left [ \ketbra{0}{0}^{\safespace'} \otimes \Ex_k (\measure \circ \auth_k)(\rho^{\msgspace \safespace \advspace}) + \ketbra{1}{1}^{\safespace'} \otimes \Ex_k \auth_k (\rho^{\msgspace \safespace \advspace}) \right].
\end{equation}
We now apply the distinguisher $D$ to this state, which acts on registers $\authspace \advspace$, and saves its output to a qubit register $\advspace'$. Because $D$ has high distinguishing advantage, applying $D$ and conditioning on $D$ giving the right answer only negligibly affects the state.  Therefore, the resulting state is $4\sqrt{2\gamma}$-close to:
\begin{equation}
\label{eq:classical_prop_2}
	\frac{1}{2} \left [ \ketbra{0}{0}^{\safespace'} \otimes \Ex_k (\measure \circ \auth_k)(\rho^{\msgspace \safespace \advspace}) \otimes \ketbra{0}{0}^{\advspace'} + \ketbra{1}{1}^{\safespace'} \otimes \Ex_k \auth_k (\rho^{\msgspace \safespace \advspace})  \otimes \ketbra{1}{1}^{\advspace'} \right].
\end{equation}
Now, this state will pass verification with probability 1, since the authentication scheme is classical.  Furthermore, since the authentication scheme is classically secure, the final state in line~\eqref{eq:classical_prop_2} can be approximated by a basis-respecting ideal adversary $\id$ acting on the state in line~\eqref{eq:classical_prop_1}:
\begin{equation}
\label{eq:classical_prop_3}
	\frac{1}{2} \left [ \ketbra{0}{0}^{\safespace'} \otimes \Ex_k  (\measure \circ \id \circ \auth_k)(\rho^{\msgspace \safespace \advspace}) + \ketbra{1}{1}^{\safespace'} \otimes \Ex_k (\id \circ \auth_k) (\rho^{\msgspace \safespace \advspace}) \right].
\end{equation}
Here, we used the fact that $\id$ and $\measure$ commute, because $\id$ is basis-respecting. Let $\tau_k^{\msgspace \safespace \advspace \advspace'}$ denote $(\id \circ \auth_k)(\rho^{\msgspace \safespace \advspace})$. Then we have that $\Ex_k  \measure (\tau_k) \approx_{2\eps} \Ex_k (\measure \circ \auth_k)(\rho^{\msgspace \safespace \advspace}) \otimes \ketbra{0}{0}^{\advspace'}$ and $\Ex_k \tau_k \approx_{2 \eps} \Ex_k \auth_k (\rho^{\msgspace \safespace \advspace})  \otimes \ketbra{1}{1}^{\advspace'}$. However this implies that $\Ex_k \tau_k^{\advspace'}$ is both $2\eps$ close to $\ketbra{0}{0}$ and $\ketbra{1}{1}$ simultaneously, which is absurd when $\eps < 1/2$ --- we have reached a contradiction. Thus, we have that if the scheme $\frac{1}{2}-4\sqrt{2\gamma}$-authenticates in the computational basis, there is no distinguisher with advantage $1-\gamma$.

\medskip

Next, we show how to boost a low-advantage distinguisher for a scheme $(\auth,\ver)$ into a high-advantage distinguisher for the product scheme $(\auth^t,\ver^t)$ which acts on message space $\msgspace^t$ by applying $\auth$ to each message component with an independent key.

A simple hybrid argument shows that, if $(\auth,\ver)$ $\eps$-authenticates in the computational basis, then $(\auth^t,\ver^t)$ $t\eps$-authenticates in the computational basis.  Note that Barnum et al.'s proof of this required somewhat more effort; however, for us, due to the fact that we consider side information in our definition, in our case the security of the product scheme comes essentially for free.

Next, assume $D$ distinguishes from measured for the state $\rho^{\msgspace \safespace \auxspace}$ in the scheme $(\auth,\ver)$ with advantage $\delta$.  Then we can boost the success probability to a distinguisher $D^t$ for the state $(\rho^{\msgspace \safespace \auxspace})^{\otimes t}$ in scheme $(\auth^t,\ver^t)$ with advantage $1-2e^{-t\delta^2/2}$.  But from the above, this means that the scheme $(\auth^t,\ver^t)$ cannot $\frac{1}{2}-8 e^{-t\delta^2/4}$-authenticate.  Thus, \[t\eps > \frac{1}{2}-8 e^{-t\delta^2/4}\]

Choosing $t=1/3\eps$ gives $\delta < 7\sqrt{\eps}$.  

\end{proof}

\subsection{Unforgeability}

In this section we show that our security definition of authentication schemes relative to a basis implies the classical security definition of authentication schemes -- namely, that the adversary, after having received the authenticated message state, cannot produce two distinct authenticated message-tag pairs with non-negligible probability. This property is called \textbf{unforgeability}. Thus this shows that our security definition recovers the Boneh-Zhandry security definition for one-time MACs.

Our model for signature forgery is the following. Let $(\auth,\ver)$ be a classical authentication scheme that is $\basis$-respecting for some basis. We will let $\basis$ be the computational basis without loss of generality. Furthermore, we will restrict our attention to MACs where for a classical message $m \in \msgspace$, $\auth_k(m) = (m,\sigma(k,m))$, although our arguments extend to general classical authentication schemes. 

Without loss of generality we can assume that the initial message state is a pure state $\ket{\rho}^{\msgspace \advspace} = \sum_m \alpha_m \ket{m}^\msgspace \otimes \ket{\varphi_m}^\advspace$ where the $\ket{\varphi_m}$ are arbitary pure states held by the adversary\footnote{For notational simplicity we omit mention of the sender/receiver's private space $\safespace$; our arguments proceed similarly when we take it into account.}. After signing, we have
$$
	\tau^{\keyspace \authspace \advspace} = \Ex_k \ketbra{k}{k} \otimes \auth_k ( \rho^{\msgspace \advspace} ).
$$
The adversary applies some superoperator $\attack$ on $\authspace \advspace$ and outputs a system on $\authspace_1 \authspace_2 \advspace$. The spaces $\authspace_1$ and $\authspace_2$ are both isomorphic to $\authspace$. Let the tampered state be denoted as
$$
	\wt{\tau}^{\keyspace \authspace_1 \authspace_2 \advspace} = \Ex_k \ketbra{k}{k} \otimes \attack(\auth_k ( \rho^{\msgspace \advspace} )).
$$
We define the \textbf{probability of forgery by $\attack$ on input $\rho$} to be the probability that, upon measuring $\keyspace$, $\authspace_1$, and $\authspace_2$ in the computational basis, we obtain a key $k$ and two valid signed messages $(m,\sigma(k,m))$ and $(m',\sigma(k,m'))$ with $m \neq m'$. 

The next theorem shows that quantum-secure authentication schemes possess the unforgeability property. The idea of the proof is as follows: suppose that there was an authentication scheme $(\auth,\ver)$, an adversary $\attack$ and an initial message state $\rho^{\msgspace}$ such that $\attack$ on input $\rho$ could forge an authenticated message with non-negligible probability. Using the fact that the authentication scheme is secure, we can in fact find a \emph{fixed} message $m \in \msgspace$ and another adversary $\what{\attack}$ that, when given an authentication of message $m$, forges a valid signed message $(m',\sigma(k,m'))$ where $m' \neq m$ with non-negligible probability. The definition of secure authentication scheme easily implies this is impossible. 

\newcommand{\forge}{{\mathsf{ForgeCheck}}}
\newcommand{\meas}{{\mathcal{M}}}

\begin{theorem}
	Let $(\auth,\ver)$ be an authentication scheme that is $\eps$-authenticating relative to the computational basis. Let $\attack$ be a forger. Then for all initial message states $\rho^{\msgspace \advspace}$, the probability of forgery by $\attack$ on input $\rho$ is at most $3\eps$.
\end{theorem}
\begin{proof}
	Suppose for contradiction that the probability of forgery is at least $\delta = 3\eps$. Since the scheme is $\eps$-authenticating relative to the computational basis, we can simulate the forger by an ideal adversary $\id$ that respects the computational basis: on input $\tau^{\keyspace \authspace \advspace}$ (the authentication of $\rho$), it first measures the $\authspace$ register to yield a valid signed message $(m,\sigma(k,m))$. Then, conditioned on this result, it applies an arbitrary quantum operation on the $\advspace$ register. Since $\attack$ is a forger, the ideal adversary $\id$ is also a forger: measuring $\keyspace \authspace \advspace$ in the computational basis will yield $k$, $(m,\sigma(k,m))$ and $(m',\sigma(k,m'))$ where $m \neq m'$ with probability at least $\delta - \eps = 2\eps$. Let $E_m$ denote the event that measuring $\authspace$ yields a valid signature of the message $m$. Let $F_m$ denote the event that measuring $\advspace$ yields a valid signature of a message that's distinct from $m$.
	
	Thus
	\begin{align*}
		\sum_{m} \Pr[E_m] \cdot \Pr[F_m | E_m] \geq 2\eps
	\end{align*}
	where the probabilities are with respect to the ideal adversary $\id$. Thus by averaging there exists an $m$ where $\Pr[F_m | E_m] \geq 2\eps$. But notice that $\Pr[E_m]$ is independent of the key, and simply $|\alpha_m|^2$, because the ideal adversary only measures the $\authspace$ register of $\tau$ in the computational basis. Thus, if we condition the state $\id(\tau)$ on the event $E_m$, we have the following state:
	$$
		\id(\tau^{\keyspace \authspace \advspace})\big |_{E_m} = \Ex_k \ketbra{k}{k}^{\keyspace} \otimes \ketbra{m,\sigma(k,m)}{m,\sigma(k,m)}^{\authspace} \otimes \id_{m,\sigma(k,m)} \left(\ketbra{\varphi_m}{\varphi_m}^{\advspace} \right)
	$$
	where $\id_{m,\sigma(k,m)}$ denotes the attack that the ideal adversary performs on the side information, conditioned on reading $(m,\sigma(k,m))$ in $\authspace$. However, $\Pr[F_m | E_m] \geq 2\eps$ implies that measuring $\Ex_k \ketbra{k}{k} \otimes \id_{m,\sigma(k,m)} \left(\ketbra{\varphi_m}{\varphi_m}^{\advspace} \right)$ in the computational basis yields $k$ and a forgery $(m',\sigma(k,m'))$ where $m' \neq m$ with probability at least $2\eps$. However, this is impossible, as $(m,\sigma(k,m))$ should have negligible information about a valid signature of $m'$. 	
\end{proof}

\section{Properties of total authentication}
\label{sec:quantum_properties}

\subsection{Encryption}

Analogous to the Barnum et al.'s~\cite{barnum2002authentication} result that authentication implies encryption, we show that authentication when considering side information must encrypt the state, even to an adversary that may be entangled with the state.  This result is compatible with Barnum et al.'s: we start from a stronger property that considers side information, and end with a stronger form of authentication that also considers side information.
%

\begin{theorem} If $(\auth,\ver)$ $\eps$-totally authenticating, then for any two states $\rho_0^{\msgspace \auxspace},\rho_1^{\msgspace  \auxspace}$ such that $\rho_0^\auxspace$ and $\rho_1^\auxspace$ are $\delta$-close, the following two states are $\delta+14\sqrt{\eps}$ close:
	\begin{itemize}
		\item $\Ex_k \auth_k (\rho_0^{\msgspace  \auxspace})$ and 
		\item $\Ex_k \auth_k (\rho_1^{\msgspace  \auxspace})$
	\end{itemize}
\end{theorem}
We remark that the conclusion here is the same as \emph{q-IND-CPA} security of quantum encryption schemes, as defined by~\cite{broadbent2015quantum}. 
\begin{proof} First, we observe that any scheme that gives $\eps$ secure encryption in the case $\delta=0$ gives $2\eps+\delta$ secure encryption in the general case. Indeed, by assumption, $\Ex_k \auth_k (\rho_0^{\msgspace  \auxspace})$ is $\eps$-close to $\Ex_k \auth_k(\ketbra{0}{0})\otimes \rho_0^\auxspace$, which is $\delta$ close to $\Ex_k \auth_k(\ketbra{0}{0})\otimes \rho_1^\auxspace$, which is $\eps$ close to $\Ex_k \auth_k(\rho_1^{\msgspace\auxspace})$. Therefore, it suffices to prove that $\auth$ is $7\sqrt{\eps}$ secure for states with $\delta=0$.
	
Our proof will closely follow the proof of Theorem~\ref{thm:indistfrommeasured}. We prove this theorem by contradiction: assuming an adversary can distinguish the two states, we will obtain a violation of the security of total authentication. Let $\rho_0^{\msgspace\auxspace},\rho_1^{\msgspace\auxspace}$ be quantum states. Suppose $D$ is a distinguisher such that
	\begin{itemize}
		\item $D(\Ex_k \auth_k  (\rho_1^{\msgspace\auxspace}))$ outputs 1 with probability at least $1-\gamma$, and 
		\item $D(\Ex_k \auth_k (\rho_0^{\msgspace\auxspace}))$ outputs 1 with probability at most $\gamma$
	\end{itemize}
	
	We set up the state $\frac{1}{2}\ketbra{0}{0}^\safespace \otimes\rho_0^{\msgspace\auxspace}+\frac{1}{2}\ketbra{1}{1}^\safespace \otimes \rho_1^{\msgspace\auxspace}$ where $\safespace$ is a private qubit register. Next, we authenticate and then apply the distinguisher $D$, saving its output to an auxiliary qubit register $\advspace'$. The subsequent state is thus $4\sqrt{2\gamma}$-close to 
\begin{equation}
\label{eq:quantum_prop_2}
	\frac{1}{2} \left [ \ketbra{0}{0}^{\safespace} \otimes \Ex_k \auth_k (\rho_0^{\msgspace \advspace}) \otimes \ketbra{0}{0}^{\advspace'} + \ketbra{1}{1}^{\safespace} \otimes \Ex_k \auth_k (\rho_1^{\msgspace\advspace})  \otimes \ketbra{1}{1}^{\advspace'} \right].
\end{equation}
	
	Now, this state will pass verification with probability 1.  Therefore, since $(\auth,\ver)$ is a secure total authentication scheme, this state is approximated by an oblivious adversary $\id$ that only acts on $\advspace$ of the authenticated state $\frac{1}{2}\ketbra{0}{0}^\safespace \otimes \auth_k(\rho_0^{\msgspace\auxspace})+\frac{1}{2}\ketbra{1}{1}^\safespace \otimes \auth_k(\rho_1^{\msgspace\auxspace})$. However, since $\rho_0^\advspace = \rho_1^\advspace$, the $\advspace$ register of this state is independent of whether the $\safespace$ qubit is $0$ or $1$. Therefore the output of the ideal adversary $\id$ must be at least $1/2$-far from the state given in line~\eqref{eq:quantum_prop_2}, which is a contradiction if $4\sqrt{2\gamma} > 1/2$. Thus, we have that if the scheme is $\frac{1}{2}-4\sqrt{2\gamma}$-totally authenticating, there is no distinguisher with advantage $1-\gamma$. Finally, we can boost a low-advantage-distinguisher for the  scheme $(\auth,\ver)$ into a high-advantage distinguisher in the same way as in Theorem~\ref{thm:indistfrommeasured}, and obtain the conclusion for all distinguishers.
	
%
%
%
%
	
\end{proof}

\subsection{Quantum Key Distribution}
\newcommand{\alice}{{\mathcal{A}}}
\newcommand{\bob}{{\mathcal{B}}}

Suppose we have a total authentication scheme. Then as argued in the Introduction, we immediately get a simple method to perform quantum key distribution. However, the QKD scheme sketched in the Introduction is rather fragile: any small amount of tampering by the adversary will cause Alice and Bob to abort. Here we sketch a slightly more robust way of carrying out QKD using a total authentication scheme.

Suppose Alice and Bob want to generate $n$ bits of perfectly correlated key bits. We now describe a protocol that takes $2$ rounds and $O(n \log n)$ bits of communication, and tolerates the adversary attacking at most $O(n/\log n)$ fraction of the qubits of communication. If this is the case, then Alice and Bob can distill at least $\Omega(n)$ bits of shared key. Let $(\auth,\ver)$ be a scheme that encodes single qubits as $O(\log n)$ qubits, and is $\eps$-totally authenticating for $\eps = n^{-\Omega(1)}$. The unitary design scheme is one such example. 

The QKD protocol is as follows:
\begin{enumerate}
\item Alice prepares the maximally entangled state over $2n$ qubits i.e. $\ket{\Phi}^{\alice \bob}=\frac{1}{\sqrt{2^n}} \sum_{x\in \{0,1\}^n}\ket{xx}^{\alice \bob}$. 
\item Alice will generate independent keys $k_1,\ldots,k_n$ for $n$ uses of the authentication scheme $(\auth,\ver)$. She authenticates each of the $n$ qubits on the $\bob$-half of $\ket{\Phi}^{\alice \bob}$ using an independent key. She sends $\bob$ to Bob.

\item Bob sends a bit to Alice acknowledging that he received some state through the quantum channel (that may have been tampered by the adversary).

\item Alice sends the keys $k_1,\ldots,k_n$ over an authenticated, but non-private, classical channel.
\item On the quantum state he received, Bob performs the verification procedure $\ver_{k_1} \otimes \cdots \ver_{k_n}$ on $n$ parts of $\log n$ qubits each. He relays to Alice which parts successfully passed verification. Let $S \subset [n]$ denote the successfully unauthenticated qubits.

\item Alice and Bob measure the part of their respective states corresponding to $S$ in the computational basis, and use these bits as their shared key.
\end{enumerate}


Since $(\auth,\ver)$ is totally authenticating, after Bob successfully unauthenticates the qubits in $S$, the qubits shared between Alice and Bob in $S$ will be $\approx \eps n$-close to the maximally entangled state. Thus when they both measure, they will both share keys $(x,x')$ that are $\eps n$-close to uniform, perfectly correlated, and private from any other system (because the maximally entangled state is in tensor product with any other quantum system). If we assume that the probability that Bob successfully verifies is not too small, then this means that Alice and Bob have successfully performed quantum key distribution.

\subsection{Key Reuse}
It is easy to see that our definition of total authentication implies that, conditioned on successful verification of an authentication scheme (satisfying total authentication), the key can be reused by the sender and receiver for some other purpose. This is because conditioned on the acceptance, the final state of the adversary is within $\eps/\alpha$ trace distance of being independent of the key, where $\alpha$ is the probability of acceptance in the authentication protocol.


\section{Quantum MACs from $3$-universal hashing}
\label{sec:threewise}

In the classical setting, secure one-time MACs can be constructed via universal hashing. Let $\{h_k\}_k$ be a strongly ($2$-)universal hash family. Then it is well known that the classical authenticiation protocol $\auth_k(m) = (m,h_k(m))$ is secure against classical adversaries~\cite{wegman1981new}. Here, we show that the \emph{same} authentication protocol is also quantum-secure, provided that the hash family $\{h_k\}_k$ satisfies the following: for all distinct $m_1,m_2,m_3$, the distribution of $(h_k(m_1),h_k(m_2),h_k(m_3))$ for a randomly chosen $k \in \keyspace$ is uniform in $\tagspace^3$. Such a family is called a \emph{$3$-universal hash family}. We will overload notation and use $k(\cdot)$ to denote the function $h_k(\cdot)$. 

We note that Boneh and Zhandry showed that, when authenticating classical messages in the one-time setting, pairwise independence is sufficient to ensure that a quantum adversary cannot forge a new signed message, as long as the length of the tag is longer than the message! When the tag is shorter than the message, they showed that pairwise independence is insecure, and $3$-wise independence is necessary. 

Our analysis of the 3-wise independent Wegman-Carter MAC requires that, in order to obtain security against quantum side information, the message tag needs to be longer than the message. Thus it is conceivable that pairwise independence is sufficient for the same guarantee; we leave this as an open question.

\begin{theorem}
\label{thm:threewise}
	Let $\keyspace = \{k \}$ be a $3$-universal hash family. Let $\auth_k(m) = (m,k(m))$ and $\ver_k$ be the corresponding verification function. Then the authentication scheme $(\auth,\ver)$ is $O(\sqrt{|\msgspace|/|\tagspace|})$-authenticating relative to the computational basis.
\end{theorem}

Before beginning the proof we first state what the implications for key length are. Suppose we wish to guarantee that the Wegman-Carter MAC is $\eps$-authenticating relative to the computational basis, then $|\msgspace|/|\tagspace| \leq O(\eps^2)$, which implies that $\log |\tagspace| \geq \log |\msgspace| + 2\log \frac{1}{\eps} + O(1)$. To ensure three-wise independence, it is sufficient for the key to have length $3\log |\msgspace| + 6\log \frac{1}{\eps} + O(1)$.

\begin{proof}
To prove this, we need to show that for all message states $\rho^{\msgspace \safespace \advspace}$ and all adversaries $\attack \in \metalinmap(\authspace \advspace,\authspace \advspace)$, the result of the QMAC is to reduce the action of the adversary on the authenticated message to an ideal, computational basis-respecting adversary. 

We will concentrate on the case of signing pure state messages -- this is because we can always purify the initial message state, and give the purification to the adversary. Furthermore, for simplicity we will consider the case where the register $\safespace$ (corresponding to the sender/receiver's private space) is empty. In other words, we will show that Wegman-Carter MAC is a quantum secure MAC when the initial message state is a state $\ket{\rho}^{\msgspace \advspace} = \sum_{m} \alpha_{m} \ket{m}^{\msgspace} \otimes \ket{\varphi_m}^{\advspace}$. The register $\msgspace$ corresponds to the message, and the register $\advspace$ is held by the adversary. At the end we will discuss how the proof generalizes to the case of non-empty $\safespace$.

It will be convenient to work with the Schmidt decomposition of $\ket{\rho}$, which we write as
$$
	\ket{\rho}^{\msgspace \advspace} = \sum_z \sqrt{\lambda_z} \left ( \sum_m \alpha_{zm} \ket{m}^{\msgspace} \right ) \otimes \ket{\varphi_z}^{\advspace}
$$
where for $z \neq z'$, we have $\ip{\varphi_z}{\varphi_{z'}} = 0$, and the $\lambda_z$'s are nonnegative numbers summing to $1$. Furthermore, the dimension of the span of $\{ \ket{\varphi_z} \}_z$ is at most $|\msgspace|$. 

After signing, the state becomes
$$
	\sigma^{\keyspace \authspace \advspace} = \Ex_{k} \ketbra{k}{k} \otimes 
	\auth_k(\rho)
$$
where $\authspace = \msgspace \tagspace$. Now consider an attack $\attack$ of the adversary. By Stinespring's Dilation Theorem, the superoperator $\attack$ can be implemented by applying a unitary $V$ on registers $\authspace \advspace$, as well as some auxiliary register $\advspace'$ held by the adversary, followed by a projective measurement $P$ on $\advspace \advspace'$, followed by tracing out $\advspace'$. 

First, we will assume that the auxiliary space $\advspace'$ is part of the purification in $\ket{\rho}^{\msgspace \advspace}$. Secondly, we will ignore the projector $P$ for now, and handle it later. 

We specify the action of $V$ on $\authspace \advspace$ as
$$
	V: \ket{m,t}^{\msgspace \tagspace} \otimes \ket{\varphi_z}^{\advspace} \mapsto \ket{\psi_{mtz}}^{\msgspace \tagspace \advspace}
$$
where $\{\ket{\psi_{mtz}}\}$ are a collection of states in $\msgspace \tagspace \advspace$ such that for all $(m,t,z) \neq (m',t',z')$, $\ip{\psi_{mtz}}{\psi_{m't'z'}} = 0$. Furthermore, write the states as follows:
$$
	\ket{\psi_{mtz}} = \sum_{a,b} \beta^{mtz}_{ab} \ket{a,b} \otimes \ket{\phi^{mtz}_{ab}}
$$
where the $\{ \ket{\phi^{mtz}_{ab}} \}$ are an arbitrary collection of unit vectors residing in the space $\advspace$, and $\ket{a,b}$ are vectors in $\authspace = \msgspace\tagspace$. Therefore after the attack we have
$$
	\wt{\sigma}^{\keyspace \authspace \advspace} = \Ex_k \ketbra{k}{k} \otimes V\, \auth_k(\rho) \, V^\dagger.
$$
Now we apply the verification procedure to this state to obtain $\tau$, where we've conditioned on the procedure accepting:
$$
	\tau^{\keyspace \authspace \advspace} = \ver(\wt{\sigma}^{\keyspace \authspace \advspace}) = \Ex_k \ketbra{k}{k} \otimes \ver_k \left(V\, \auth_k(\rho) \, V^\dagger \right)
$$
Note that $\tau$ does not have unit trace in general (because the verification procedure $\ver_k$ may not pass with probability $1$). For a fixed key $k$, we can write 
$$
	\ket{\tau_{k}} = \ver_k \, V\, \auth_k \ket{\rho} = \sum_{z,m,a} \sqrt{\lambda_z} \alpha_{zm} \beta^{mk_mz}_{ak_a} \,\, \ket{a}^{\msgspace} \otimes \ket{\phi^{mk_mz}_{ak_a} }^{\tagspace}
$$
where we abbreviate $k(m)$ and $k(a)$ by $k_m$ and $k_a$ respectively. We can decompose the vector $\ket{\tau_k} = \ket{\tau_{k,ideal}} + \ket{\tau_{k,err}}$ where
\begin{align}
	\ket{\tau_{k,ideal}}^{\msgspace \tagspace \advspace} = \sum_{z,m} \sqrt{\lambda_z} \alpha_{zm} \beta^{mk_mz}_{mk_m}  \,\, \ket{m}^{\msgspace}\otimes \ket{\phi^{mk_mz}_{mk_m} }^{\advspace} \label{eq:ideal} \\
	\ket{\tau_{k,err}}^{\msgspace \tagspace \advspace} = \sum_{z,m,a: a \neq m} \sqrt{\lambda_z}  \alpha_{zm} \beta^{mk_mz}_{ak_a}  \,\, \ket{a}^{\msgspace} \otimes \ket{\phi^{mk_mz}_{ak_a} }^{\advspace}
	\label{eq:error}
\end{align}
Thus $\tau^{\keyspace \authspace \advspace} = \tau_{ideal} + \tau_{err}$ where $$\tau_{ideal} = \Ex_k \ketbra{k}{k} \otimes \ketbra{\tau_{k,ideal}}{\tau_{k,ideal}},$$ and let $$\tau_{err} = \Ex_k \ketbra{k}{k} \otimes \left ( \ketbra{\tau_{k,ideal}}{\tau_{k,err}} + \ketbra{\tau_{k,err}}{\tau_{k,ideal}} + \ketbra{\tau_{k,err}}{\tau_{k,err}} \right).$$
The $\tau_{ideal}$ represents the part of $\tau$ that looks like it underwent an \emph{ideal} attack, while the term $\tau_{ideal}$ represents the rest of $\tau$. We will bound this error term and show that its size is small within $\tau$, and thus this will show that $\tau$ is close to the result of an ideal attack. 

To bound the size of $\tau_{err}$, we note that 
\begin{align*}
	\| \tau_{err} \|_1 &\leq \Ex_k \left [ 2 \| \,\, \ketbra{\tau_{k,ideal}}{\tau_{k,err}} \,\,  \|_1 + \| \,\, \ketbra{\tau_{k,err}}{\tau_{k,err}} \,\, \|_1 \right ] \\
	&= \Ex_k \left [ 2 \sqrt{\ip{\tau_{k,err}}{\tau_{k,err}} \cdot \ip{\tau_{k,ideal}}{\tau_{k,ideal}}} + \ip{\tau_{k,err}}{\tau_{k,err}} \right ] \\
	&\leq 3 \Ex_k \sqrt{\ip{\tau_{k,err}}{\tau_{k,err}}} \\
	&\leq 3 \sqrt{ \Ex_k \ip{\tau_{k,err}}{\tau_{k,err}} }
\end{align*}
where in the equality we used that for two pure states $\ket{\varphi}$ and $\ket{\psi}$, $\| \,\, \ketbra{\varphi}{\psi} \,\, \|_1 = \sqrt{ \ip{\varphi}{\varphi} \cdot \ip{\psi}{\psi}}$. In the second-to-last inequality we used that $\ip{\tau_{k,ideal}}{\tau_{k,ideal}} \leq 1$, and in the last inequality we used the concavity of the square-root function. Now,
\begin{align}
	\Ex_k \ip{\tau_{k,err}}{\tau_{k,err}} &=  \Ex_k \sum_{\substack{z,z' \\ a,m,m' : a \notin \{m,m' \}}} \sqrt{\lambda_z \lambda_{z'}} \cdot \alpha_{zm} \conj{\alpha}_{z'm'} \cdot \beta^{mk_mz}_{ak_a} \conj{\beta}^{m'k_{m'}z'}_{ak_a} \cdot \ip{\phi^{m'k_{m'}z'}_{ak_a}}{\phi^{mk_{m}z}_{ak_a}}\\
	&= \sum_{\substack{z,z' \\ a,m,m' : a \notin \{m,m' \}}} \sqrt{\lambda_z \lambda_{z'}} \cdot \alpha_{zm} \conj{\alpha}_{z'm'} \cdot \left (\Ex_k \beta^{mk_mz}_{ak_a} \conj{\beta}^{m'k_{m'}z'}_{ak_a} \cdot \ip{\phi^{m'k_{m'}z'}_{ak_a}}{\phi^{mk_{m}z}_{ak_a}} \right) \label{eq:3wise_1}
\end{align}

Observe that, for every $a, m,m'$ such that $a \notin \{m,m'\}$, $k_a$ is independent of $k_m$ and $k_{m'}$ (this is where we use $3$-wise independence of $k$). Therefore, we can write
\begin{align*}
\Ex_k \beta^{mk_mz}_{ak_a} \conj{\beta}^{m'k_{m'}z'}_{ak_a} \cdot \ip{\phi^{m'k_{m'}z'}_{ak_a}}{\phi^{mk_{m}z}_{ak_a}} = \Ex_{k,h} \beta^{mk_mz}_{ah_a} \conj{\beta}^{m'k_{m'}z'}_{ah_a} \cdot \ip{\phi^{m'k_{m'}z'}_{ah_a}}{\phi^{mk_{m}z}_{ah_a}}
\end{align*}
where the expectation on the right hand side is over two independent hash families $k$ and $h$. We have equality because $(k_m,k_{m'},k_a)$ and $(k_m,k_{m'},h_a)$ are identically distributed.

This motivates us to define 
\begin{align*}
	\xi_1 = \Ex_{k,h}  \sum_{z,z',m,m'} \sqrt{\lambda_z \lambda_{z'}} \cdot \alpha_{zm} \conj{\alpha}_{z'm'} \cdot \beta^{mk_mz}_{mh_m} \conj{\beta}^{m'k_{m'}z'}_{mh_m} \cdot \ip{\phi^{m'k_{m'}z'}_{mh_m}}{\phi^{mk_{m}z}_{mh_m}} \\
	\xi_2 = \Ex_{k,h}  \sum_{z,z',m} \sqrt{\lambda_z \lambda_{z'}} \cdot \alpha_{zm} \conj{\alpha}_{z'm} \cdot \beta^{mk_mz}_{mh_m} \conj{\beta}^{mk_{m}z'}_{mh_m} \cdot \ip{\phi^{mk_{m}z'}_{mh_m}}{\phi^{mk_{m}z}_{mh_m}}.
\end{align*}
We will momentarily show that $\xi_1$ and $\xi_2$ are small in magnitude. Assuming this, we add $\xi_1$ and $\xi_2$ to~\eqref{eq:3wise_1} to get a nicer-looking sum:
\begin{align}
	\text{\eqref{eq:3wise_1} } + \xi_1 + \conj{\xi}_1 - \xi_2 &= \sum_{\substack{z,z' \\ a,m,m'}} \sqrt{\lambda_z \lambda_{z'}} \cdot \alpha_{zm} \conj{\alpha}_{z'm'} \cdot \left (\Ex_{k,h} \beta^{mk_mz}_{ah_a} \conj{\beta}^{m'k_{m'}z'}_{ah_a} \cdot \ip{\phi^{m'k_{m'}z'}_{ah_a}}{\phi^{mk_{m}z}_{ah_a}} \right) \\
	&= \frac{1}{|\tagspace|} \sum_{z,z',m,m'} \sqrt{\lambda_z \lambda_{z'}} \cdot \alpha_{zm} \conj{\alpha}_{z'm'} \cdot \Ex_{k} \sum_{a,b} \beta^{mk_mz}_{ab} \conj{\beta}^{m'k_{m'}z'}_{ab} \cdot \ip{\phi^{m'k_{m'}z'}_{ab}}{\phi^{mk_{m}z}_{ab}} \\
	&= \frac{1}{|\tagspace|}  \sum_{z,m} \lambda_z \cdot |\alpha_{zm}|^2 \Ex_{k} \sum_{a,b} |\beta^{mk_mz}_{ab}|^2 \\
	&= \frac{1}{|\tagspace|}.
\end{align}
To go from the second line to the third line we used the orthogonality conditions 
$$\ip{\psi_{m't'z'}}{\psi_{mtz}} = \sum_{a,b} \beta^{mtz}_{ab} \conj{\beta}^{m't'z'}_{ab} \ip{\phi^{m't'z'}_{ab}}{\phi^{mtz}_{ab} } = 0$$ whenever $(m,t,z) \neq (m',t',z')$.

Now we bound the magnitudes of $\xi_1$ and $\xi_2$. We use Cauchy-Schwarz repeatedly to bound $|\xi_1|$:
\begin{align}
	|\xi_1| &= \frac{1}{|\tagspace|} \left | \Ex_k \sum_{z,z',m,m'} \sqrt{\lambda_z \lambda_{z'}} \cdot \alpha_{zm} \conj{\alpha}_{z'm'} \cdot \sum_b \beta^{mk_mz}_{mb} \conj{\beta}^{m'k_{m'}z'}_{mb} \cdot \ip{\phi^{m'k_{m'}z'}_{mb}}{\phi^{mk_{m}z}_{mb}} \right | \\
		&\leq \frac{1}{|\tagspace|} \Ex_k \sqrt {\sum_{z,z',m,m'} |\alpha_{zm}|^2 |\alpha_{z'm'}|^2 \cdot \left | \sum_b \beta^{mk_mz}_{mb} \conj{\beta}^{m'k_{m'}z'}_{mb} \cdot \ip{\phi^{m'k_{m'}z'}_{mb}}{\phi^{mk_{m}z}_{mb}}  \right|^2 } \\
		&\leq \frac{1}{|\tagspace|} \Ex_k \sqrt {\sum_{z,z',m,m'} |\alpha_{zm}|^2 |\alpha_{z'm'}|^2 \cdot \left ( \sum_b |\beta^{mk_mz}_{mb}|^2 \cdot \| \ket{\phi^{mk_{m}z}_{mb}}  \|^2 \right)\left ( \sum_b |\beta^{m'k_{m'}z'}_{mb}|^2 \cdot \| \ket{\phi^{m'k_{m'}z'}_{mb}}  \|^2 \right) } \\
		&\leq \frac{1}{|\tagspace|} \Ex_k \sqrt {\sum_{z,z',m,m'} |\alpha_{zm}|^2 |\alpha_{z'm'}|^2 \cdot \left ( \sum_{a,b} |\beta^{mk_mz}_{ab}|^2 \cdot \| \ket{\phi^{mk_{m}z}_{ab}}  \|^2 \right)\left ( \sum_{a,b} |\beta^{m'k_{m'}z'}_{ab}|^2 \cdot \| \ket{\phi^{m'k_{m'}z'}_{ab}}  \|^2 \right) } \\
		&\leq \frac{1}{|\tagspace|} \Ex_k \sqrt {\sum_{z,z',m,m'} |\alpha_{zm}|^2 |\alpha_{z'm'}|^2 \cdot 2 } \\
		&\leq \sqrt{2} \frac{|\msgspace|}{|\tagspace|}.
\end{align}
In the last line, we used the fact that the dimension of the span of the $\ket{\varphi_z}$'s is at most $|\msgspace|$. A similar calculation will show that $|\xi_2| \leq \sqrt{2}|\msgspace|/|\tagspace|$ as well. Putting everything together, we get that 
$$
	\left | \Ex_k \ip{\tau_{k,err}}{\tau_{k,err}} \right | \leq (1 + 3\sqrt{2})|\msgspace|/|\tagspace|.
$$
This implies that
$$
	\| \tau - \tau_{ideal} \|_1 \leq 3\sqrt{6|\msgspace|/|\tagspace|}.
$$
Recall that we have ignored the final projector $P$ that the real adversary $\attack$ may have applied after applying the unitary $V$. Since $P$ acts on $\advspace$ only, it commutes with the verification operation, and thus we have that
$$
	\| P \tau P^\dagger - P \tau_{ideal}P^\dagger \|_1 \leq 3\sqrt{6|\msgspace|/|\tagspace|}.
$$
where $P \tau P^\dagger = \Ex_k \ketbra{k}{k} \otimes \ver_k \circ \attack \circ \auth_k(\rho^{\msgspace \advspace})$, the true final state of the protocol. Applying $\auth_k^{-1}$ on L.H.S. doesn't increase the distance.

Finally, we have to argue that $P \tau_{ideal} P^\dagger$ is actually equal to $\Ex_k \ketbra{k}{k} \otimes \Pi_{\mathcal{V}_k} \circ \id \circ \auth_k(\rho^{\msgspace \advspace})$ for some computational basis-respecting adversary $\id$. The ideal adversary behaves as follows when given the $\authspace \advspace$ registers of $\sigma^{\keyspace \authspace \advspace}$:
\begin{enumerate}
	\item The adversary prepares auxiliary registers $\msgspace' \tagspace' \advspace_2$ in the $\ket{0 \cdots 0}$ state. The $\authspace' = \msgspace' \tagspace'$ registers are isomorphic to $\msgspace \tagspace$, and $\advspace_2$ is a qubit register.
	\item First the ideal adversary makes a copy of the $\msgspace \tagspace$ registers in the computational basis and coherently stores the copy in auxiliary registers $\msgspace' \tagspace'$.
	\item The ideal adversary then applies the original adversary unitary $V$ to registers $\msgspace' \tagspace' \advspace$.
	\item The adversary checks whether the values of the $\msgspace \tagspace$ and $\msgspace' \tagspace'$ registers are the same in the computational basis; if so, the $\advspace_2$ qubit is set to $\ket{0}$, and the $\msgspace' \tagspace'$ registers are set to $\ket{0 \cdots 0}$. Otherwise, it is kept at $\ket{1}$. In other words, the basis vector $\ket{m,t,m',t',0}^{\msgspace \tagspace \msgspace' \tagspace' \advspace_2}$ is mapped to 
	$\ket{m,t,0 \cdots 0}^{\msgspace \tagspace \msgspace' \tagspace' \advspace_2}$ iff $m = m'$ and $t = t'$. 	
	
	
	\item The adversary measures the $\advspace_2$ qubit register, and the $\advspace$ register using the POVM element $\{P, \Id - P\}$, and accepts only on outcome $\ket{0}$ for $\advspace_2$ and $P$ for $\advspace$.
\end{enumerate}
Observe that this ideal attack $\id$ can be implemented as
$$
	\id : \sigma^{\authspace \auxspace} \mapsto \Tr_{\authspace' \advspace_2} \left ( (P \otimes \ketbra{0}{0}^{\advspace_2}) V_{ideal} \sigma^{\authspace \advspace} V_{ideal}^\dagger (P \otimes \ketbra{0}{0}^{\advspace_2}) \right )
$$
where $V_{ideal}$ is an isometry mapping the space $\authspace\advspace$ to the space $\authspace \authspace'\advspace  \advspace_2$, $P \otimes \ketbra{0}{0}^{\advspace_2}$ is a projector acting on $\advspace\advspace_2$, and $\Tr_{\authspace' \advspace_2} (\cdot)$ is the partial trace over system $\authspace' \advspace_2$. Furthermore, $V_{ideal}$ is an isometry that leaves the $\msgspace \tagspace$ registers unchanged, and hence is a computational basis-respecting adversary. Observe that $P \tau_{ideal}^{\keyspace \authspace \advspace} P^\dagger = \Pi_{\mathcal{V}_k}\left(\id(\sigma^{\keyspace \authspace \advspace}) \right)$.

Thus we have shown how the final state of the protocol can be simulated by an ideal adversary, when the input state is of the form $\ket{\rho^{\msgspace \advspace}}$. What about the case when the state is of the form $\ket{\rho^{\msgspace \safespace \advspace}}$? The previous analysis works exactly the same, where we bundle together $\safespace$ and $\advspace$ as the adversary space, but we make the simple observation that here the ideal adversary constructed above would act as the identity on $\safespace$. 

Thus we have established the simulation of every adversary by an ideal adversary for \emph{every} input state $\ket{\rho^{\msgspace \safespace \advspace}}$, so this implies that $(\auth,\ver)$ is $O(\sqrt{M/T})$-authenticating relative to the computational basis.

\end{proof}


\section{Total authentication (with key leakage) from complementary classical authentication}
\label{sec:qft}

In the previous section, we saw how the classical Wegman-Carter message authentication scheme is still secure even when used on a superposition of messages, and even if the adversary has access to quantum side information about the messages. Here, we will show that using the Wegman-Carter scheme as a primitive, we obtain \emph{total quantum state authentication}, which implies encryption of the quantum state. 

The quantum state authentication scheme is simple: the sender authenticates the message state using the Wegman-Carter MAC in the computational basis, and then authenticates again in the Fourier basis (using a new key). The verification procedure is the reverse of this: the receiver first checks the outer authentication, performs the inverse Fourier transform, and then checks the inner authentication. We call this the ``Auth-QFT-Auth'' scheme. This is pleasingly analogous to the quantum one-time pad (QOTP), which encrypts quantum data using the classical one-time pad in complementary bases. However, the QOTP does not have authentication properties. Our analysis requires the 3-wise independence property of the Wegman-Carter MAC.

There is one slight caveat: we show that Auth-QFT-Auth achieves total authentication \emph{with key leakage}. That is, we argue that conditioned on the receiver verification succeeding, the effect of an arbitrary adversary is to have ignored the authenticated state, and only act on the adversary's side information, in a manner that may depend on the key used for the second authentication (what we call the ``outer key''). In other words, we sacrifice the secrecy of the outer key, but in exchange we get complete quantum state encryption. 

\subsection{The Auth-QFT-Auth scheme}

Let $\ket{\rho}^{\msgspace \advspace} = \sum_m \alpha_m \ket{m}^{\msgspace} \otimes \ket{\varphi_m}^{\advspace}$ be the initial message state, where $\advspace$ is held by the adversary. Just like in the proof of Theorem~\ref{thm:threewise}, we will omit mention of the sender/receiver's private space $\safespace$, and discuss how our proof generalizes to the case of non-empty $\safespace$ later.

Again, it will be advantageous to rewrite this state in terms of the Schmidt decomposition:  
$$
	\ket{\rho}^{\msgspace \advspace} = \sum_z \sqrt{\lambda_z} \left ( \sum_m \alpha_{zm} \ket{m}^{\msgspace} \right ) \otimes \ket{\varphi_z}^{\advspace}
$$
where for $z \neq z'$, we have $\ip{\varphi_z}{\varphi_{z'}} = 0$, and the $\lambda_z$'s are nonnegative numbers summing to $1$. Furthermore, the dimension of the span of $\{ \ket{\varphi_z} \}_z$ is at most $|\msgspace|$. 

The authentication scheme is the composed operation $\auth_2 ( H^{\otimes N} (\auth_1 (\rho)))$, where $\auth_1$ is the \emph{inner} authentication scheme that uses key $k$, $H^{\otimes N}$ is the quantum Fourier transform over $\Z_2$, and $\auth_2$ is the \emph{outer} authentication that uses key $h$. The keys $k$ and $h$ are independent.

The inner authentication scheme $\auth_1$ maps $\msgspace$ to $\authspace_1 = \msgspace \tagspace_1$. We define $N = |\authspace_1|$. $H$ is the single-qubit Hadamard unitary, and the Fourier transform $H^{\otimes N}$ acts on $\authspace_1$. The outer authentication scheme $\auth_2$ maps $\authspace_1$ to $\authspace_2 = \msgspace \tagspace_1 \tagspace_2$. The keys $k$ and $h$ live in the registers $\keyspace$ and $\hilb$, respectively. The evolution of the initial message state is as follows:
\begin{enumerate}
	\item \textbf{Inner authentication}. When the inner authentication key (henceforth called the \emph{inner key}) is $k$, the state becomes $$\sum_z \sqrt{\lambda_z} \left( \sum_m \alpha_{zm} \ket{m,k(m)}^{\authspace_1} \right) \otimes \ket{\varphi_z}^{\advspace}$$
	
	\item \textbf{Fourier transform over $\Z_2$}: Let $\{ \ket{x} \}$ be a basis for $\authspace_1$. Then: 
	$$\frac{1}{\sqrt{N}} \sum_z \sqrt{\lambda_z} \left( \sum_{m,x} \alpha_{zm} (-1)^{(m,k(m)) \cdot x} \ket{x}^{\authspace_1} \right) \otimes \ket{\varphi_z}^{\advspace}.$$
	
	\item \textbf{Outer authentication}. The outer key is denoted by $h$. The final authenticated state is then 
	$$\ket{\sigma_{kh}}^{\authspace \tagspace_2 \advspace} = \frac{1}{\sqrt{N}} \sum_z \sqrt{\lambda_z} \left( \sum_{m,x} \alpha_{zm} (-1)^{(m,k(m)) \cdot x} \ket{x,h(x)}^{\authspace_1 \tagspace_2} \right) \otimes \ket{\varphi_z}^{\advspace}$$
	where $\tagspace_2$ is the space of the tag $h(x)$.
\end{enumerate}

Let 
$$
	\sigma^{\keyspace \hilb \authspace_1 \tagspace_2 \advspace} = \Ex_{kh} \ketbra{kh}{kh}^{\keyspace \hilb} \otimes \ketbra{\sigma_{kh}}{\sigma_{kh}}^{\authspace_1 \tagspace_2 \advspace}.
$$
The adversary is then given the $\authspace_1 \tagspace_2$ registers of $\sigma$, and performs a general unitary attack $V$ that acts on $\authspace_1 \tagspace_2 \advspace$:
$$
	\wt{\sigma}^{\keyspace \hilb \authspace_1 \tagspace_2 \advspace} = V \sigma V^\dagger.
$$
Let $\wt{\tau}^{\keyspace \hilb \msgspace \advspace} =  \ver_1 \circ \QFT^{-1}   \circ \ver_2 (\wt{\sigma})$.
Let the inner authentication scheme be the 3-wise independent hashing QMAC with tag length $\log T$, and message length $\log M$. Let the outer authentication scheme be a QMAC that $\eps$-authenticates with respect to the computational basis. 

The Auth-QFT-Auth scheme can potentially leak some bits of the outer key $h$, but we will show that this is the \emph{only} thing that is leaked; otherwise, it is performs total authentication (and hence encryption).
	
\begin{theorem}[Security of the Auth-QFT-Auth scheme]
The Auth-QFT-Auth scheme is $\delta$-totally authenticating with outer key leakage, where $\delta = \eps + O(\sqrt{|\msgspace|^{3/2}/|\tagspace_1|})$.	
\end{theorem}

Again before starting the proof we consider the key requirements. The outer authentication scheme need not be a Wegman-Carter MAC, but let's assume that it is. In order to achieve $\delta$-total authentication, the inner MAC must be such that $|\msgspace|^{3/2}/|\tagspace_1| \leq O(\delta^2)$, or in other words, $\log |\tagspace_1| \geq \frac{3}{2} \log |\msgspace| + 2\log \frac{1}{\delta} + O(1)$. The key needed for the inner MAC must be at least $\frac{9}{2} \log |\msgspace| + 6\log \frac{1}{\delta} + O(1)$. The ``message length'' that is given to the outer MAC is $\log |\msgspace| + \log |\tagspace_1| \geq \frac{5}{2} \log |\msgspace| + 2\log \frac{1}{\delta} + O(1)$, and thus $\log |\tagspace_2| \geq \frac{5}{2} \log |\msgspace| + 4\log \frac{1}{\delta} + O(1)$. The key length for the outer MAC needs to be at least $\frac{15}{2} \log |\msgspace| + 12\log \frac{1}{\delta} + O(1)$, so the total key needed is $12\log |\msgspace| + 18\log \frac{1}{\delta} + O(1)$. 

While the inner key can be recycled (upon successful verification), the outer key unfortunately cannot be.

\begin{proof}
We will let $M = |\msgspace|$, $T = |\tagspace_1|$, and $N = MT = |\authspace_1|$. We will assume that $M^{3/2} \leq T$; otherwise the theorem statement is vacuous. 

Suppose the outer authentication scheme was $\eps$-secure. By definition, there exists an ideal computational basis adversary $\id$ such that $\| \ver_2(\wt{\sigma}) - \ver_2(\id(\sigma)) \|_1 \leq \eps$, where $\ver_2$ denotes the verification procedure for the outer authentication scheme. There exists a computational basis-respecting linear map $\Lambda \in \linmap(\authspace_2 \advspace)$ such that
$$
	\id: \sigma \mapsto \Lambda \sigma \Lambda^\dagger.
$$
Since $\Lambda$ is computational basis-respecting, we have for all $(x,s,z)$:
$$
	\Lambda \ket{x,s}^{\authspace_1 \tagspace_2} \otimes \ket{\varphi_z}^{\advspace} = \ket{x,s}^{\authspace_1 \tagspace_2} \otimes \ket{\phi_{xsz}}^{\advspace}.
$$
for some collection of (not necessarily normalized) states $\{ \ket{\phi_{xsz}} \}$.

Therefore the effect of the adversary on the authenticated state (after verification) is to be close to $\id(\sigma) = \Ex_{k,h} \ketbra{kh}{kh} \otimes \ketbra{\tau_{kh}}{\tau_{kh}}$ where for fixed inner/outer keys $k,h$
$$
	\ket{\tau_{kh}} = \frac{1}{\sqrt(N)} \sum_z \sqrt{\lambda_z} \sum_{m,x} \alpha_{zm} (-1)^{(m,k(m)) \cdot x} \ket{x} \otimes \ket{\phi_{xh_xz}}.
$$
Thus, the final state that Bob has, after performing full (i.e. inner and outer) verification, is $\eps$-close to 
$$
	\Ex_{k,h} \ketbra{kh}{kh} \otimes \ketbra{\mu_{kh}}{\mu_{kh}}
$$
where 
$$
	\ket{\mu_{kh}} = \sum_z \sqrt{\lambda_z} \sum_{m} \left( \frac{1}{N} \sum_{x,m'} \alpha_{zm'} (-1)^{(m + m',k(m) + k(m')) \cdot x} \right) \ket{m} \otimes \ket{\phi_{xh_xz}}.
$$
Then security of Auth-QFT-Auth is established if we show that for every $h$,
$$
	\Ex_k \left \| \ket{\mu_{kh}} - \ket{\nu_h} \right \|^2
$$
is small, where 
$$
	\ket{\nu_h}^{\msgspace \advspace} = \sum_z \sqrt{\lambda_z} \sum_m \alpha_{zm} \ket{m}^{\msgspace} \otimes \ket{\eta_{hz}}^{\advspace}
$$
with $\ket{\eta_{hz}}^\advspace = \frac{1}{N} \sum_x \ket{\phi_{xh_xz}}^{\advspace}$. Assuming this, the next Lemma will show that there is an ideal oblivious, but outer key-dependent, adversary whose actions lead to the global state $\Ex_{kh} \ketbra{kh}{kh} \otimes \ketbra{\nu_h}{\nu_h}$.

\begin{lemma}[Constructing the ideal oblivious adversary]
\label{lem:ideal_adv}
	For all $h$ there exists an ideal oblivious adversary $\id_h$ acting on $\advspace$ only such that
	$$
		\ketbra{\nu_h}{\nu_h}^{\msgspace \advspace} = \id_h(\ketbra{\rho}{\rho}^{\msgspace \advspace}).
	$$
\end{lemma}
\begin{proof}
We now construct an ideal adversary $\id_h$, derived from the computational basis adversary $\id$. By definition of $\id$, there exists a computational basis-respecting isometry $V \in \isometry(\authspace_2 \advspace,\authspace_2 \advspace \authspace_2' \advspace_2)$ where $\authspace_2'$ is an auxiliary register isomorphic to $\authspace_2$, and $\advspace_2$ is an auxiliary qubit register, such that
$$
	\id : \sigma^{\authspace \auxspace} \mapsto \Tr_{\authspace' \advspace_2} \left ( \Pi V \sigma^{\authspace \advspace} V^\dagger \Pi \right ).
$$
Here $\Pi = P \otimes \ketbra{0}{0}^{\advspace_2}$ for some projector $P$ acting on $\advspace$. Furthermore, $V$ is computational basis respecting:
$$
	\Pi V\ket{x,s}^{\authspace_2} \otimes \ket{\varphi_z}^{\advspace} = \ket{x,s}^{\authspace_2} \otimes \ket{\phi_{xsz}}^{\advspace} \otimes \ket{0 \cdots 0}^{\authspace_2' \advspace_2}
$$
where the $\ket{\phi_{xsz}}^{\advspace}$ were defined above.

\newcommand{\aspace}{{\mathcal{A}}}

Now we construct the ideal general adversary $\id_h$ as follows:
\begin{enumerate}
	\item First, the adversary creates the entangled state $\ket{\Phi_h}^{\aspace \aspace'} = \frac{1}{\sqrt{N}} \sum_x \ket{x,h(x)}^{\aspace} \ket{x,h(x)}^{\aspace'}$ in new registers $\aspace \otimes \aspace'$, which are isomorphic to $\authspace_2 \otimes \authspace_2$, and $\{ \ket{x} \}$ is a basis for $\authspace_1$.
	
	\item It then applies the unitary $V$ to half of $\ket{\Phi_h}^{\aspace \aspace'}$ that resides in $\aspace$, and the $\advspace$ part of the input state $\ket{\rho}$. 
	\item The adversary measures $\aspace \aspace' \advspace \advspace_2$ using the projective measurement $\{Q, \Id - Q\}$, where $Q = \ketbra{\Phi_h}{\Phi_h}^{\aspace \aspace'} \otimes \Pi$. The adversary discards the outcome corresponding to $\Id - Q$, and leaves the state unnormalized:
	$$
	\frac{1}{N} \sum_{z,x,m} \sqrt{\lambda_z} \alpha_{zm} \ket{m}^{\msgspace} \, \ket{\phi_{xsz}}^{\advspace} \ket{\Phi}^{\aspace \aspace'} \ket{0 \cdots 0}^{\authspace_2' \advspace_2}
	$$
	\item The adversary discards the $\aspace \aspace' \authspace_2' \advspace_2$ registers:
$$
	\frac{1}{N} \sum_{z,x,m} \sqrt{\lambda_z} \alpha_{zm} \ket{m}^{\msgspace} \otimes \ket{\phi_{xsz}}^{\advspace}
	$$
\end{enumerate}
This is precisely the state $\ket{\nu_h}$, and the $\id_h$ only interacts with $\advspace$ and auxiliary registers in the adversary's control, so it is an ideal general adversary.
\end{proof}

We now turn to bounding $\Ex_k \left \|  \ket{\mu_{kh}} - \ket{\nu_h} \right \|^2$:
\begin{align*}
&\Ex_k \left \|  \ket{\mu_{kh}} - \ket{\nu_h} \right \|^2 \\
&= \frac{1}{N^2} \Ex_k \sum_{m,z,z'} \sqrt{\lambda_z \lambda_{z'}} \sum_{\substack{x',x'',m',m'' \\ m \notin \{ m',m''\}}} \conj{\alpha}_{z'm'}\alpha_{zm''} (-1)^{(m + m',k(m) + k(m')) \cdot x'}  (-1)^{(m + m'',k(m) + k(m'')) \cdot x''} \ip{\phi_{x'z'}^h}{\phi_{x''z}^h} \\
&= \frac{1}{N^2} \sum_{\substack{m,z,z' \\ x',x'',m',m''\\ m \notin \{ m',m''\}}} \sqrt{\lambda_z \lambda_{z'}} \conj{\alpha}_{z'm'}\alpha_{zm''} (-1)^{(m + m') \cdot x_1' + (m + m'') \cdot x_1''} \ip{\phi_{x'z'}^h}{\phi_{x''z}^h} \Ex_k (-1)^{(k(m) + k(m')) \cdot x'_2} (-1)^{(k(m) + k(m'')) \cdot x_2''}.
\end{align*}
We use the abbreviation $\ket{\phi^h_{xz}} = \ket{\phi_{xh_xz}}$. In the second line, we divided $x$ into two parts $(x_1,x_2)$, where $x_1$ corresponds to $\msgspace$, and $x_2$ corresponds to $\tagspace_1$. We focus on the expectation $\chi_{m,m',m'',x_2',x_2''} = \Ex_k (-1)^{(k(m) + k(m')) \cdot x'_2} (-1)^{(k(m) + k(m'')) \cdot x_2''}$. We consider two cases:

\paragraph{Case 1: $m' = m'', m' \neq m$.} Then $\chi_{m,m',m'',x_2',x_2''} = 0$ if $x_2' \neq x_2''$, otherwise $\chi_{m,m',m'',x_2',x_2''} = 1$.
\begin{align*}
	&\frac{1}{N^2} \left | \sum_{\substack{z,z', x',x'' \\ m,m': m\neq m'}}\sqrt{\lambda_z \lambda_{z'}} \conj{\alpha}_{z'm'}\alpha_{zm'} (-1)^{(m + m') \cdot (x_1' + x_1'')} \ip{\phi_{x'z'}^h}{\phi_{x''z}^h} \chi_{m,m',m',x',x''} \right | \\
	&=\frac{1}{N^2} \left | \sum_{\substack{z,z', x',x_1''\\ m,m': m\neq m'}}\sqrt{\lambda_z \lambda_{z'}} \conj{\alpha}_{z'm'}\alpha_{zm'} (-1)^{(m + m') \cdot (x_1' + x_1'')} \ip{\phi_{x'z'}^h}{\phi_{x_1''x_2'z}^h} \right | \\
&\leq \frac{1}{N^2} \sum_{z,z'} \sqrt{\lambda_z \lambda_{z'}} \sqrt{\sum_{m \neq m'} \left | \sum_{x',x_1''} (-1)^{(m+ m') \cdot (x_1' + x_1'')}  \ip{\phi_{x'z'}^h}{\phi_{x_1''x_2'z}^h} \right |^2} \qquad \qquad \text{(Cauchy-Schwarz)}\\
&\leq \frac{1}{N^2} \sum_{z,z'} \sqrt{\lambda_z \lambda_{z'}} \sqrt{\sum_{m, m'} \sum_{x',x_1'',\wt{x}',\wt{x}_1''} (-1)^{(m+ m') \cdot (x_1' + x_1'' + \wt{x}_1' + \wt{x}_1'')}  \ip{\phi_{\wt{x}_1''\wt{x}_2' z}^h}{\phi_{\wt{x}'z'}^h} \ip{\phi_{x'z'}^h}{\phi_{x_1''x_2'z}^h} } \\
&= \frac{1}{N^2} \sum_{z,z'} \sqrt{\lambda_z \lambda_{z'}} \sqrt{M^2 \sum_{\substack{x',x_1'',\wt{x}',\wt{x}_1'' \\ x_1' + x_1'' + \wt{x}_1' + \wt{x}_1'' = 0}} \ip{\phi_{\wt{x}_1''\wt{x}_2'z}^h}{\phi_{\wt{x}'z'}^h} \ip{\phi_{x'z'}^h}{\phi_{x_1''x_2'z}^h} } \\
&\leq \frac{1}{N^2} \sum_{z,z'} \sqrt{\lambda_z \lambda_{z'}} \sqrt{M^3 N^2} \\
&\leq \frac{M^{5/2}}{N} \qquad \qquad \text{(at most $M$ $z$'s)} \\
&= \frac{M^{3/2}}{T}.
\end{align*}

\paragraph{Case 2: $m,m',m''$ are all distinct.} Then $\chi_{m,m',m'',x_2',x_2''} = 0$ unless $x_2' = x_2'' = 0$, in which case $\chi_{m,m',m'',x_2',x_2''} = 1$. This uses the three-independence of $k(\cdot)$.
\begin{align*}
	&\frac{1}{N^2} \left | \sum_{\substack{z,z', x',x'' \\ m,m',m''\text{ distinct}}}\sqrt{\lambda_z \lambda_{z'}} \conj{\alpha}_{z'm'}\alpha_{zm'} (-1)^{(m + m') \cdot x_1' + (m + m'') \cdot x_1''} \ip{\phi_{x'z'}^h}{\phi_{x''z}^h} \chi_{m,m',m'',x',x''} \right | \\
	&=\frac{1}{N^2} \left | \sum_{\substack{z,z', x_1',x_1'' \\ m,m',m''\text{ distinct}}}\sqrt{\lambda_z \lambda_{z'}} \conj{\alpha}_{z'm'}\alpha_{zm'} (-1)^{(m + m') \cdot x_1' + (m + m'') \cdot x_1''} \ip{\phi_{x_1'0z'}^h}{\phi_{x_1''0z}^h} \right | \\
	&\leq \frac{1}{N^2} \sum_{z,z'} \sqrt{\lambda_z \lambda_{z'}} \sqrt{\sum_{m,m',m''\text{ distinct}} \left | \sum_{x_1',x_1''} (-1)^{(m + m') \cdot x_1' + (m + m'') \cdot x_1''} \ip{\phi_{x_1'0z'}^h}{\phi_{x_1''0z}^h} \right |^2} \qquad \qquad \text{(Cauchy-Schwarz)}\\
	&\leq \frac{M^{9/2}}{N^2} \\
	&\leq \frac{M^{3/2}}{T}
\end{align*}
where we used the fact that $M^{3/2} \leq T$. Therefore, for every $h$ we have
$$
\Ex_k \left \|  \ket{\mu_{kh}} - \ket{\nu_h} \right \|^2 = O(M^{3/2}/T)
$$
as desired. Using Fact~\ref{fact:pure_to_dens} and Jensen's inequality, $ \Ex_{kh} \left \|  \ketbra{\mu_{kh}}{\mu_{kh}} - \ketbra{\nu_h}{\nu_h} \right \| \leq O(\sqrt{M^{3/2}/T})$.

Thus, the final state of Bob is $\eps + O(\sqrt{M^{3/2}/T})$-close to 
\begin{align*}
	\Ex_{kh} \ketbra{kh}{kh} \otimes \ketbra{\nu_h}{\nu_h} = \Ex_{kh} \ketbra{kh}{kh} \otimes \id_h(\ketbra{\rho}{\rho})
\end{align*}
where $\id_h$ are the ideal adversaries given by Lemma~\ref{lem:ideal_adv}.

To conclude the theorem, we now observe that when $\safespace$ is non-empty, we can use the same analysis as above where we bundle together $\safespace$ and $\advspace$ as a joint adversary register, and the ideal adversary given by Lemma~\ref{lem:ideal_adv} will act as the identity on the $\safespace$ register. This establishes that $(\auth,\ver)$ is a total authentication scheme with outer key leakage.
\end{proof}

\section{Total authentication from approximate unitary designs}
\label{sec:design}

We now present a scheme that satisfies the strongest security definition, that of total authentication (without \emph{any} key leakage). In particular, this implies complete reuse of the entire key. This property of complete reuse of the key was not known before; it is not known whether the entire key can be reused in the authentication scheme of Barnum, et al~\cite{barnum2002authentication}.


This scheme is based on \emph{unitary designs}, which are in some sense the quantum analogue of $t$-wise independent hash functions: a $t$-unitary design (also simply called a \emph{$t$-design}) is a distribution $\mathscr{D}$ over unitary matrices such that degree $t$ polynomials cannot distinguish between a unitary drawn from $\mathscr{D}$ and a fully random unitary. Furthermore, there are constructions of efficient unitary designs~\cite{brandao2012local}. 

\subsection{The unitary design scheme}

We call this scheme the \emph{unitary design scheme}.  Let $s$ be a security parameter. The input state is  $\ket{\rho}^{\msgspace \advspace}$, where the $\advspace$ register is held by the adversary.

\begin{enumerate}
	\item The sender Alice first appends $s$ $\ket{0}$ qubits in an auxiliary $\tagspace$ register.
	\item Using her secret key $k$, Alice samples a random unitary $U_k$ drawn from an (approximate) unitary $t$-design that acts jointly on $\msgspace \otimes \tagspace$. We will set the parameter $t = 8$.
	
	\item Alice applies $U_k$ to the $\msgspace \otimes \tagspace$ register, and sends $\msgspace \otimes \tagspace$ across the quantum channel to Bob.
	
	\item Bob receives some state, and applies the inverse unitary $U_k^\dagger$ to it. He measures the last $s$ qubits and accepts if they all measure to be $0$. Otherwise he rejects.
\end{enumerate}

\begin{theorem}
\label{thm:design}
	The unitary design scheme is efficiently computable, and is $2^{-s/2}$-totally authenticating.
\end{theorem}


This scheme is inspired by the \emph{Clifford code authentication scheme}, first proposed by Aharonov, et al.~\cite{aharonov2008interactive}, and further analyzed in~\cite{dupuis2012actively,broadbent2016efficient}. Our protocol is exactly the same, except the ensemble of unitaries, instead of being an approximate $8$-design, is the \emph{Clifford group}, which is a well-studied set of unitaries that are central to quantum error-correction, simulation, and more. It was also recently shown that the Clifford group is a $3$-unitary design~\cite{webb2015clifford,zhu2015multiqubit}\footnote{However, it is not an $8$-design}.~\cite{dupuis2012actively,broadbent2016efficient} show that the Clifford authentication scheme is secure even against entangled adversaries; however, as mentioned before, their security guarantee does not take into account the key.

Our unitary design scheme is also very similar to the \emph{non-malleable quantum encryption scheme} proposed by Ambainis, Bouda, and Winter~\cite{ambainis2009nonmalleable}, wherein a unitary $2$-design is used to encrypt a quantum state. However, non-malleable quantum encryption does not imply authentication. 


We now remark upon the key requirements of the unitary design scheme. Constructions of approximate unitary $8$-designs acting on $n$ qubits involve choosing a random quantum circuit of size $\Theta(n^2)$, and thus the randomness required is $\Theta(n^2)$~\cite{brandao2012local}. This asymptotically matches the randomness requirements required of the Clifford scheme described above, but is much larger than the randomness requirements of the purity-testing-based protocol of~\cite{hayden2011universal}, which uses $\Theta(n)$ bits of key to authenticate an $n$-qubit quantum state.


\paragraph{Notation and useful lemmas.} We set up some notation. We let $\msgspace$ denote the message space, $\tagspace$ to denote the space of the dummy zero qubits. We let $\authspace = \msgspace \otimes \tagspace$. We let $M = |\msgspace|$, $|\tagspace| = 2^s$, and $N = M2^s = |\authspace|$. 

Let $\attack$ be an adversary acting on $\authspace \otimes \advspace$. By the Stinespring representation theorem, there exists a unitary $V$ acting on a possibly larger space $\authspace \otimes \advspace \otimes \advspace'$, followed by a projection $P$ that acts on $\advspace \advspace'$, followed by a partial trace over $\advspace'$. However without loss of generality we shall simply treat this additional space $\advspace'$ as part of $\advspace$, and ignore the partial trace operation. Thus, the adversary's action is to perform some unitary $V$ on $\authspace \otimes \advspace$, followed by a projection on $P$ on $\advspace$.

To analyze the behavior of this scheme, we will first analyze the case when the randomizing unitary $U$ is drawn from the Haar measure over the unitary group $U(\authspace)$, rather from a $t$-design. We will show that this scheme is totally authenticating. Then, we will show that actually using a $O(1)$-unitary design will suffice.

The crucial hammer we will need is a version of \emph{Levy's Lemma}: 

\begin{definition}
A function $f: U(d) \to \R$ is $\eta$-Lipschitz if
$$
	\sup_{U_1,U_2 \in U(d)} \frac{ | f(U_1) - f(U_2) |}{\| U_1 - U_2 \|_2} \leq \eta.
$$
\end{definition}

\begin{lemma}[Levy's Lemma~\cite{milman2009asymptotic}]
	Let $f : U(d) \to \R$ be an $\eta$-Lipschitz function on the unitary group of dimension $d$ with mean $\Ex f$. Then
	$$
		\Pr \big ( \big | f - \Ex f \big | \geq \delta \big ) \leq 4 \exp \left ( -\frac{C d \delta^2}{\eta^2} \right)
	$$
	where $C = 2/9\pi^3$ and the probability is over $U$ drawn from the Haar measure on $U(d)$.
\end{lemma}

Another useful lemma we will need is the following, giving two formulas for averaging over the (Haar measure of the) unitary group. We use $\delta_{ij}$ to denote the Dirac delta function that is $1$ if $i = j$ and $0$ otherwise. 

\begin{lemma}[Appendix B.5 of~\cite{beenakker1997random}]
\label{lem:haar_moments}
For a function $f : \unitary(d) \to \R$, we let $\langle f \rangle$ to denote $\int f(U) \,\, dU$, where $\int \cdot dU$ is integration over the Haar measure on $\unitary(d)$. Then 
	\begin{align*}
		\langle U_{ab} U_{ij} U^*_{a'b'} U^*_{i'j'} \rangle &= \frac{1}{d^2 - 1} ( \delta_{aa'} \delta_{bb'} \delta_{ii'} \delta_{jj'} + \delta_{ai'} \delta_{bj'} \delta_{ia'} \delta_{jb'}) \\ &\qquad \qquad -\frac{1}{d(d^2-1)} ( \delta_{aa'} \delta_{bj'} \delta_{ii'} \delta_{jb'} + \delta_{ai'} \delta_{bb'} \delta_{ia'} \delta_{jj'})
	\end{align*}

\end{lemma}

\subsection{Total authentication with Haar-random unitaries}
We now prove that the unitary design scheme yields total authentication. Let $\Lambda_U = \bra{0}^{\otimes s} U^\dagger V U \ket{0}^{\otimes s}$ is be a map from $\msgspace \otimes \advspace$ to $\msgspace \otimes \advspace$. 


\begin{lemma}
\label{lem:side_inf_concentration}
	Let $N = \dim(\authspace)$. For all $\delta > 0$, for all initial message states $\ket{\rho}^{\msgspace \advspace}$ have that
	$$
		\Pr_U \left ( \| \Gamma_V \ket{\rho} - \Lambda_U \ket{\rho} \|_2^2 \geq 2^{-s} + \delta \right) \leq \exp(-C'N \delta^2)
	$$
	where $\Gamma_V = \Tr_{\authspace}(V)/\dim(\authspace)$, $C'$ is a universal constant, and $U$ is a Haar-random unitary.
\end{lemma}

\begin{remark}
Again, as before, we will omit mention of the sender/receiver's private space $\safespace$, and treat it as part of the adversary's private space $\advspace$. The ideal adversary we construct will act as the identity on $\safespace$.
\end{remark}

\begin{proof}
	First, we write $\ket{\rho}^{\msgspace \authspace} = \sum_x \rho_x \ket{x}^{\msgspace} \otimes \ket{\varphi_x}^{\advspace}$ where $\{ \ket{x} \}$ is a basis for $\msgspace$, and $\{ \ket{\varphi_x} \}$ are arbitrary unit vectors in $\advspace$. 
	
	Write $U$ as the following:
	$$
		U = \sum_{u,x} \ketbra{\psi_{u,x}}{u,x}
	$$
	where $\ket{u} \in \tagspace, \ket{x} \in \msgspace$ are standard basis vectors, and $\{\ket{\psi_{u,x}}\} \subset \tagspace \otimes \msgspace$ is a set of orthonormal unit vectors. Then $U \ket{0}^{\otimes s}$ becomes a linear operator that accepts vectors in $\msgspace$ and outputs vectors in $\authspace = \tagspace \otimes \msgspace$: 
		$$
		U\ket{0}^{\otimes s} = \sum_{x} \ketbra{\psi_{0^s,x}}{x}
	$$
	We will simply write $\ket{\psi_x}$ to denote $\ket{\psi_{0^s,x}}$. We can write $\Lambda_U$ as
	$$
\Lambda_U = \sum_{x,x'} \ketbra{x'}{x} \,\, \bra{\psi_{x'}} V \ket{\psi_x}.	
$$	
	
Let's compute the average operator
 \begin{align}
 		\int \Lambda_U \,\, dU &= 		\sum_{x,x'}  \ketbra{x}{x'}\,\, \int  \bra{\psi_{x}} V \ket{\psi_{x'}} \,\, dU \\
		&= \sum_{x}  \ketbra{x}{x}\,\, \int  \bra{\psi_{x}} V \ket{\psi_x} \,\, dU  +\sum_{x \neq x'}  \ketbra{x}{x'}\,\, \int  \bra{\psi_{x}} V \ket{\psi_{x'}} \,\, dU \label{eq:design1} \\
		&= \sum_x \ketbra{x}{x} \otimes \frac{1}{\dim(\authspace)} \Tr_{\authspace}(V) \\
		&= \Id^{\msgspace} \otimes \Gamma_V
 \end{align}
 The second term in~\eqref{eq:design1} (the sum over off-diagonal elements) averages to $0$, because for $x \neq x'$, the vectors $\ket{\psi_{x'}}$ and $\ket{\psi_x}$ are random orthogonal unit vectors. Conditioned on a fixing of $\ket{\psi_x}$, for any vector $\ket{\varphi}$ that is orthogonal to $\ket{\psi_x}$, $\ket{\psi_{x'}}$ is equally likely to be $\ket{\varphi}$ or $-\ket{\varphi}$, so $\int  \bra{\psi_{x'}} V \ket{\psi_x} \,\, dU = 0$.

	In the last step we used the fact that given an operator $X$ mapping $\authspace \otimes \advspace$ to $\authspace \otimes \advspace$, if we average over the unit sphere, $\int (\bra{\psi}^{\authspace} \otimes \Id^{\advspace}) X (\ket{\psi}^{\authspace} \otimes \Id^{\advspace}) \,\, d\psi$ is equal to the partial trace $\Tr_{\authspace}(X)/\dim(\authspace)$. We'll let $N$ denote $\dim(\authspace)$.
 
 Thus, this tells us that on average, this operator should act as the identity on $\msgspace$ and some linear map (not necessarily unitary) $\Gamma_V$ on $\advspace$. We now prove that $\Lambda_U$ behaves this way on $\ket{\rho}$ with high probability. Define
 $$
 	f(U) =\| \Gamma_V \ket{\rho} - \Lambda_U \ket{\rho} \|_2^2. 
 $$
 
 \paragraph{Bounding the average of $f$.}
 Expanding $f$ and averaging, we get
 \begin{align}
 	\int f(U) \,\, dU &= \int \big( \bra{\rho}\Gamma_V^\dagger - \bra{\rho}\Lambda_U^\dagger \big)\big( \Gamma_V \ket{\rho} - \Lambda_U \ket{\rho} \big)\,\, dU \\
	&= \int \bra{\rho} \Gamma_V^\dagger \Gamma_V \ket{\rho} - \bra{\rho}\Lambda_U^\dagger \Gamma_V \ket{\rho} - \bra{\rho}\Gamma_V^\dagger \Lambda_U \ket{\rho} + \bra{\rho} \Lambda_U^\dagger \Lambda_U \ket{\rho} \,\, dU \\
	&= -\bra{\rho} \Gamma_V^\dagger \Gamma_V \ket{\rho} + \int \bra{\rho} \Lambda_U^\dagger \Lambda_U \ket{\rho} \,\, dU
 \end{align}
 where in the last line we used our calculation of $\int \Lambda_U \,\, dU$ above. We bound this last term. We have that
 $$
 	\Lambda_U \ket{\rho} = \sum_{x,x'} \rho_{x'} \ket{x} \,\, \bra{\psi_{x}} V (\ket{\psi_{x'}} \otimes \ket{\varphi_{x'}} )
 $$
 Thus
 \begin{align}
\int \bra{\rho} \Lambda_U^\dagger \Lambda_U \ket{\rho} \,\, dU &= \int \sum_{x,x',x''} \rho_{x'} \rho_{x''}^* (\bra{\psi_{x''}} \otimes \bra{\varphi_{x''}}) V^\dagger \ket{\psi_{x}}  \bra{\psi_{x}} V (\ket{\psi_{x'}} \otimes \ket{\varphi_{x'}}) \,\, dU \\
&= \sum_{x'} |\rho_{x'}|^2 \sum_x \int (\bra{\psi_{x'}} \otimes \bra{\varphi_{x'}}) V^\dagger \ket{\psi_{x}}  \bra{\psi_{x}} V (\ket{\psi_{x'}} \otimes \ket{\varphi_{x'}}) \,\, dU \\
&= \sum_{x'} |\rho_{x'}|^2 \sum_x \int \big \| \bra{\psi_{x}} V (\ket{\psi_{x'}} \otimes \ket{\varphi_{x'}}) \big \|^2_2 \,\, dU \\
&= \sum_{x'} |\rho_{x'}|^2 \sum_{x \neq x'} \int \big \| \bra{\psi_{x}} V (\ket{\psi_{x'}} \otimes \ket{\varphi_{x'}}) \big \|^2_2 \,\, dU + \\
&\qquad \qquad \qquad \sum_{x} |\rho_{x}|^2 \int \big \| \bra{\psi_{x}} V (\ket{\psi_{x}} \otimes \ket{\varphi_{x}}) \big \|^2_2 \,\, dU.
\label{eq:lambda_avg}
 \end{align}
 Let $\{ \ket{z} \}$ be a basis for $\advspace$. Now notice that 
\begin{align}
\big \| \bra{\psi_{x}} V (\ket{\psi_{x'}} \otimes \ket{\varphi_{x'}}) \big \|^2_2 &= 
\big \| \sum_z \ketbra{z}{z}^{\advspace} \,\, \bra{\psi_{x}} V (\ket{\psi_{x'}} \otimes \ket{\varphi_{x'}}) \big \|^2_2 \\
&= \sum_z \big | (\bra{\psi_{x}} \otimes \bra{z}) V (\ket{\psi_{x'}} \otimes \ket{\varphi_{x'}}) \big |^2.
\end{align}
Write $\ket{\varphi_x'} = \sum_z \beta_{x' z} \ket{z}$. Then we have
\begin{align}
	\big | (\bra{\psi_{x}} \otimes \bra{z}) V (\ket{\psi_{x'}} \otimes \ket{\varphi_{x'}}) \big |^2 &= \left | \sum_{z'} \beta_{x' z'} (\bra{\psi_{x}} \otimes \bra{z}) V (\ket{\psi_{x'}} \otimes \ket{z'}) \right |^2 \\
	&= \left | \sum_{z'} \beta_{x' z'} \sum_{ij} V_{(i,z),(j,z')} U_{ix}^* U_{jx'} \right |^2 \\
	&= \sum_{z',z''} \beta_{x' z'} \beta_{x' z''}^* \sum_{iji'j'} V_{(i,z),(j,z')} V_{(i',z),(j',z'')}^* U_{i'x} U_{jx'}  U_{ix}^*  U_{j'x'}^* 
	\label{eq:expansion}
\end{align}
 where the rows and columns of $V$ are indexed by $(i,z)$ and $(j,z')$, respectively. Again, we identify $\ket{\psi_x}$ as the $x$'th column of $U$, and $U_{ix}$ denotes the $i$'th entry of $\ket{\psi_x}$.
 
We now go back to bound the sum over $x \neq x'$ in~\eqref{eq:lambda_avg}. Fix $x,x'$ such that $x \neq x'$. Substituting~\eqref{eq:expansion} in and using Lemma~\ref{lem:haar_moments}, we get:
\begin{align}
&\int \big \| \bra{\psi_{x'}} V (\ket{\psi_x} \otimes \ket{\varphi_{x'}}) \big \|^2_2 \,\, dU \\
&= \int \sum_{z,z',z''} \beta_{x' z'} \beta_{x' z''}^* \sum_{iji'j'} V_{(i,z),(j,z')} V_{(i',z),(j',z'')}^* U_{i'x} U_{jx'}  U_{ix}^*  U_{j'x'}^* 
  \,\, dU \\
  &= \sum_{z',z''} \beta_{x' z'} \beta_{x' z''}^* \left [ \frac{1}{N^2 - 1} \sum_{ijz} V_{(i,z),(j,z')} V_{(i,z),(j,z'')}^*  - \frac{1}{N(N^2  - 1)} \sum_{ii'z} V_{(i,z),(i,z')} V_{(i',z),(i',z'')}^* \right ] \\
  &= \frac{N}{N^2 - 1}  - \frac{1}{N(N^2  - 1)} \sum_{z',z''} \beta_{x' z'} \beta_{x' z''}^* \sum_{ii'z} V_{(i,z),(i,z')} V_{(i',z),(i',z'')}^* \\
  &= \frac{N}{N^2 - 1}  - \frac{1}{N(N^2  - 1)} \sum_z \left | \sum_{z',i} \beta_{x',z'} V_{(i,z),(i,z')} \right|^2 \\
  &\leq \frac{N}{N^2 - 1}
  \label{eq:expansion2}
\end{align}
where we used the fact that $V$ is unitary and that $\sum_{z'} |\beta_{x' z'}|^2 = 1$. Summing~\eqref{eq:expansion2} over all $x \neq x'$, we get
$$
	 \sum_{x'} |\rho_{x'}|^2 \sum_{x \neq x'} \int \big \| \bra{\psi_{x}} V (\ket{\psi_{x'}} \otimes \ket{\varphi_{x'}}) \big \|^2_2 \,\, dU \leq \sum_{x'} |\rho_{x'}|^2 \sum_{x \neq x'} \frac{N}{N^2 - 1} = \frac{N(M-1)}{N^2 - 1}.
$$
Now fix an $x$; we bound the second term of~\eqref{eq:lambda_avg}. Using Lemma~\ref{lem:haar_moments} again, we have
\begin{align}
&\int \big \| \bra{\psi_{x}} V (\ket{\psi_{x}} \otimes \ket{\varphi_{x}}) \big \|^2_2 \,\, dU \\
&=\sum_{z,z',z''} \beta_{x z'} \beta_{x z''}^* \int \sum_{iji'j'} V_{(i,z),(j,z')} V_{(i',z),(j',z'')}^*  U_{i'x} U_{jx'}  U_{ix}^*  U_{j'x'}^*  \,\, dU \\
&= \left [\frac{1}{N^2 - 1}  - \frac{1}{N(N^2  - 1)} \right] \cdot \left [ \sum_z \left | \sum_{z',i} \beta_{x',z'} V_{(i,z),(i,z')} \right|^2 + N \right ] \\
&= \frac{1}{N(N+1)} \cdot \left [ \sum_z \left | \sum_{z',i} \beta_{x',z'} V_{(i,z),(i,z')} \right|^2 + N \right ] 
\end{align}
Putting everything together, we can bound~\eqref{eq:lambda_avg} by
\begin{align}
	\int \bra{\rho} \Lambda_U^\dagger \Lambda_U \ket{\rho} \,\, dU &\leq \frac{N(M-1)}{N^2 - 1} + \frac{1}{N(N+1)} \cdot \left [ \sum_z \left | \sum_{z',i} \beta_{x',z'} V_{(i,z),(i,z')} \right|^2 + N \right ] \\
	&= \frac{NM - 1}{N^2 - 1} + \frac{1}{N(N+1)} \sum_z \left | \sum_{z',i} \beta_{x',z'} V_{(i,z),(i,z')} \right|^2.
\end{align}
We have to compare this to $\bra{\rho} \Gamma_V^\dagger \Gamma_V \ket{\rho} = \| \Gamma_V \ket{\rho} \|^2_2$. We expand $\Gamma_V \ket{\rho}$:
\begin{align}
\Gamma_V \ket{\rho}  &= \frac{1}{N}\Tr_\authspace(V) \ket{\rho} \\
&= \frac{1}{N} \sum_{i,z} \ket{z}^{\advspace} \bra{i,z}V\ket{i} \left ( \sum_{x,z'} \rho_x \beta_{xz'} \ket{x}^\msgspace \otimes \ket{z'}^{\advspace} \right) \\
&= \frac{1}{N} \sum_{x,z} \rho_x \ket{x,z}^{\msgspace \advspace} \left ( \sum_{i,z'} \beta_{xz'}  \bra{i,z} V \ket{i,z'} \right) \\
&= \frac{1}{N} \sum_{x,z} \rho_x \ket{x,z}^{\msgspace \advspace} \left ( \sum_{i,z'} \beta_{xz'} V_{(i,z),(i,z')} \right)
\end{align}
So therefore 
$$
\bra{\rho} \Gamma_V^\dagger \Gamma_V \ket{\rho} = \frac{1}{N^2} \sum_z \left | \sum_{z',i} \beta_{x',z'} V_{(i,z),(i,z')} \right|^2.
$$
This shows that our desired average of $f$ is small:
$$
	\int f(U) \,\, dU \leq \frac{NM - 1}{N^2 - 1}.
$$

\paragraph{Bounding the Lipschitz constant of $f$.} We compute the Lipschitz continuity of $f$ in parts. Let $g(U) = \bra{\rho}\Gamma_V^\dagger \Lambda_U \ket{\rho}$, where $\ket{\rho} = \sum_x \rho_x \ket{x} \otimes \ket{\varphi_x}$. Expanding, we get
\begin{align}
	g(U) &= \bra{\rho} (\Id^{\authspace} \otimes \Gamma_V^\dagger) \sum_{x,x'} \ketbra{x}{x'} \otimes \bra{\psi_x}V \ket{\psi_{x'}} \ket{\rho} \\
	&= \sum_{x,x'} \rho_x^* \rho_{x'} (\bra{\psi_x} \otimes \bra{\varphi_x}) \Gamma_V^\dagger V (\ket{\psi_{x'}}\otimes \ket{\varphi_{x'}}) \\
	&= \left ( \sum_x \rho_x^* \bra{\psi_x} \otimes \bra{\varphi_x} \right) \Gamma_V^\dagger V \left (\sum_x \rho_x \ket{\psi_{x}}\otimes \ket{\varphi_{x}} \right) \\
	&= \bra{\theta} \Gamma_V^\dagger V \ket{\theta}
\end{align}
where we used that $\Gamma_V^\dagger$ is an operator that acts on $\advspace$ only, and we define $\ket{\theta} = \sum_x \rho_x \ket{\psi_{x}}\otimes \ket{\varphi_{x}}$. Thus for two unitaries $U, \what{U}$, we have
\begin{align}
	\big | g(U) - g(\what{U}) \big | &= \big | \bra{\theta} \Gamma_V^\dagger V \ket{\theta} - \bra{\what{\theta}} \Gamma_V^\dagger V \ket{\what{\theta}} \big | \\
	&= \left | \Tr \left ( \Gamma_V^\dagger V  ( \ketbra{\theta}{\theta} - \ketbra{\what{\theta}}{\what{\theta}} ) \right ) \right | \\
	&\leq \left \| \Gamma_V^\dagger V \right\|_\infty \cdot \big \| \, \ketbra{\theta}{\theta} - \ketbra{\what{\theta}}{\what{\theta}}  \, \big \|_1
\end{align}
where in the inequality we used H\"{o}lder's inequality for matrices: $\Tr(AB) \leq \|A \|_\infty \|B\|_1$. Now, the operator norm is submultiplicative, so $\| \Gamma_V^\dagger V \|_\infty \leq \| \Gamma_V^\dagger \|_\infty \cdot \| V \|_\infty \leq \| \Gamma_V^\dagger \|_\infty$, because $V$ is a unitary and hence its operator norm is $1$. But then $\|\Gamma_V^\dagger \|_\infty = \frac{1}{N} \| \sum_y \bra{y} V \ket{y} \|_\infty \leq \frac{1}{N} \sum_y \| \bra{y} V \ket{y} \|_\infty$, because the operator norm satisfies the triangle inequality. Here, $\ket{y}$ is a basis element of $\authspace$, and $\bra{y}V \ket{y}$ is an operator that maps $\advspace$ to $\advspace$. For each $y$, we can bound $\| \bra{y} V \ket{y} \|_\infty \leq 1$. This implies that $| g(U) - g(\what{U}) | \leq \| \ketbra{\theta}{\theta} - \ketbra{\what{\theta}}{\what{\theta}}  \|_1$.

Thus the Lipschitz constant of $g$ can be bounded by
	$$
		\eta_g \leq \sup_{U,\what{U}} \frac{ 2 \| \ket{\theta} - \ket{\what{\theta}} \|_2}{ \| U - \what{U} \|_2}.
	$$
	Since the columns of $U, \what{U}$ are $\ket{\psi_{u,x}}$ and $\ket{\what{\psi}_{u,x}}$, the denominator $\| U - \what{U} \|_2$ can be written as $\sqrt{\sum_{u,x} \| \ket{\psi_{u,x}} - \ket{\what{\psi}_{u,x}} \|_2^2 }$. Notice that the numerator only depends on the column vectors $\ket{\psi_{0^s,x}} = \ket{\psi_x}$ and $\ket{\what{\psi}_{0^s,x}} = \ket{\what{\psi}_x}$, so the denominator can be minimized to be $\sqrt{\sum_{x} \| \ket{\psi_{x}} - \ket{\what{\psi}_{x}} \|_2^2 }$ without affecting the numerator. The numerator can be bounded as
\begin{align}
	 \| \sum_x \rho_x (\ket{\psi_x} \otimes \ket{\varphi_x} - \ket{\what{\psi}_x} \otimes \ket{\varphi_x}) \|_2 &\leq \sum_x |\rho_x| \cdot \|\ket{\psi_x} \otimes \ket{\varphi_x} - \ket{\what{\psi}_x} \otimes \ket{\varphi_x}\|_2 \\
		&\leq \sqrt{ \sum_x |\rho_x|^2 \sum_x \|\ket{\psi_x} - \ket{\what{\psi}_x} \|_2^2} \\
		&= \sqrt{ \sum_x \|\ket{\psi_x} - \ket{\what{\psi}_x} \|_2^2}
\end{align}
where in the first line we used the triangle inequality, and in the second line we used Cauchy-Schwarz. Thus the Lipschitz constant of $g$ is at most $2$.


Now we bound the Lipschitz continuity of $h(U) = \bra{\rho}\Lambda_U^\dagger \Lambda_U \ket{\rho}$. We have that
	\begin{align}
		h(U) &= \sum_{x,x',x''} \rho_{x'} \rho_{x''}^*  (\bra{\psi_{x''}} \otimes \bra{\varphi_{x''}}) V^\dagger \ket{\psi_{x}}  \bra{\psi_{x}} V (\ket{\psi_{x'}} \otimes \ket{\varphi_{x'}} )  \\
		&=  \sum_x \bra{\theta} V^\dagger \ketbra{\psi_x}{\psi_x} V \ket{\theta} \\
		&= \Tr \left ( \sum_x  \ketbra{\psi_x}{\psi_x} V \ketbra{\theta}{\theta} V^\dagger \right)
	\end{align}
	where $\ket{\theta}$ is the same as above. Let $\Pi_U = \sum_x  \ketbra{\psi_x}{\psi_x}$. Therefore 
	\begin{align}
		| h(U) - h(\what{U}) | &= \left | \Tr \left ( \Pi_U V \ketbra{\theta}{\theta} V^\dagger - \Pi_{\what{U}} V \ketbra{\what{\theta}}{\what{\theta}} V^\dagger \right) \right | \\
		&= \left | \Tr \left ( \Pi_U V \left ( \ketbra{\theta}{\theta} - \ketbra{\what{\theta}}{\what{\theta}} \right) V^\dagger + (\Pi_U - \Pi_{\what{U}}) V \ketbra{\what{\theta}}{\what{\theta}} V^\dagger \right) \right | \\
		&\leq \left \| \ketbra{\theta}{\theta} - \ketbra{\what{\theta}}{\what{\theta}}  \right \|_1 + \left | \bra{\what{\theta}} V^\dagger ( \Pi_U - \Pi_{\what{U}} ) V \ket{\what{\theta}} \right | \\
		&\leq \left \| \ketbra{\theta}{\theta} - \ketbra{\what{\theta}}{\what{\theta}} \right \|_1 + \left \| \Pi_U - \Pi_{\what{U}}  \right \|_\infty
	\end{align}
	where in the first inequality we use that $\Pi_U = \Pi_U \cdot \Pi_U$ is a projector, and that $\Tr(\Pi X) \leq \| X \|_1$ for all operators $X$ and $-\Id \leq \Pi \leq \Id$. The second term can be bounded by
	\begin{align}
	\left \| \Pi_U - \Pi_{\what{U}}  \right \|_\infty &=\left \| \sum_x \ketbra{\what{\psi}_x}{\psi_x} + (\ket{\psi_x} - \ket{\what{\psi}_x})\bra{\psi_x} - \ketbra{\what{\psi}_x}{\what{\psi}_x} \right \|_\infty  \\
	&=  \left \| \sum_x (\ket{\psi_x} - \ket{\what{\psi}_x})\bra{\psi_x} + \ket{\what{\psi}_x}(\bra{\psi_x} - \bra{\what{\psi}_x} ) \right \|_\infty \\
	&\leq \left \| \sum_x (\ket{\psi_x} - \ket{\what{\psi}_x})\bra{\psi_x}  \right\|_\infty + \left \| \sum_x (\ket{\psi_x} - \ket{\what{\psi}_x})\bra{\what{\psi}_x}  \right\|_\infty\\ &=  \sup_{\ket{v}} \left \| \sum_x (\ket{\psi_x} - \ket{\what{\psi}_x})\ip{\psi_x}{v}  \right\|_2 + \sup_{\ket{v}} \left \| \sum_x (\ket{\psi_x} - \ket{\what{\psi}_x})\ip{\what{\psi}_x}{v}  \right\|_2\\
	&\leq \sup_{\ket{v}}  \sum_x | \ip{\psi_x}{v}| \cdot \left \| \ket{\psi_x} - \ket{\what{\psi}_x} \right \|_2  +\sup_{\ket{v}}  \sum_x | \ip{\what{\psi}_x}{v}| \cdot \left \| \ket{\psi_x} - \ket{\what{\psi}_x} \right \|_2 \\
	&\leq \sup_{\ket{v}}  \sqrt{ \sum_x |\ip{\psi_x}{v}|^2 \sum_x \left \| \ket{\psi_x} - \ket{\what{\psi}_x} \right \|_2^2 } + \sup_{\ket{v}} \sqrt{ \sum_x |\ip{\what{\psi}_x}{v}|^2 \sum_x \left \| \ket{\psi_x} - \ket{\what{\psi}_x} \right \|_2^2 }  \\
	&\leq 2\sqrt{\sum_x \left \| \ket{\psi_x} - \ket{\what{\psi}_x} \right \|_2^2 }.
	\end{align}
	Therefore the Lipschitz constant $\eta_h$ of $h$ is at most $4$, so the Lipschitz constant $\eta$ of $f$ is at most $8$.

Now we invoke Levy's Lemma once more, and we obtain 
	\begin{align}
		\Pr \big ( |\| \Gamma_V \ket{\rho} - \Lambda_U \ket{\rho} \|_2^2-\int f(U) \,\, dU| \geq \delta \big ) &\leq 4 \exp \left ( -\frac{C N \delta^2}{\eta^2} \right) \\
		&\leq 4 \exp \left ( -C' M^2/N \right)
	\end{align}
where $\delta = 2M/N$ and $C'$ is some universal constant.

\end{proof}

\subsection{Constructing the ideal oblivious adversary}

Now we demonstrate that the map $\ket{\rho}^{\msgspace \advspace} \mapsto \Gamma_V \ket{\rho}^{\msgspace \advspace}$ can be implemented by an ideal oblivious adversary. 

Consider the following ideal adversary, which given a state $\ket{\sigma}^{\authspace \advspace}$ performs the following:
\begin{enumerate}
	\item First, the adversary creates a maximally entangled state $\ket{\Phi}^{\authspace' \authspace''} = \frac{1}{\sqrt{N}} \sum_y \ket{yy}^{\authspace' \authspace''}$ in new registers $\authspace' \otimes \authspace''$. 
	\item It then applies the unitary $V$ to half of $\ket{\Phi}^{\authspace' \authspace''}$ that resides in $\authspace'$, and the $\advspace$ part of $\ket{\sigma}^{\authspace \advspace}$. The state currently looks like:
	\begin{align}
		&\frac{1}{\sqrt{N}} \sum_{y} (\Id^{\authspace} \otimes V^{ \advspace \authspace'})\ket{\sigma}^{\authspace \advspace} \otimes \ket{yy}^{\authspace' \authspace''} \\
		&= \frac{1}{\sqrt{N}} \sum_{y} (\Id^{\authspace} \otimes \Id^{\authspace' \authspace''} \otimes V^{ \advspace \authspace'})\ket{\sigma}^{\authspace \advspace} \otimes \ket{yy}^{\authspace' \authspace''} \\
		&=\frac{1}{\sqrt{N}} \sum_{y,y'} (\Id^{\authspace} \otimes \bra{y'}^{\authspace'} V^{\advspace \authspace'}\ket{y}^{\authspace'})\ket{\sigma}^{\authspace \advspace} \otimes \ket{y'y}^{\authspace' \authspace''}
	\end{align}
	\item The adversary projects $\authspace' \authspace''$ using the projector $\ketbra{\Phi}{\Phi}^{\authspace' \authspace''}$ (and leaves the state unnormalized):
	$$
	\frac{1}{N} \sum_{y} (\Id^{\authspace} \otimes \bra{y}^{\authspace'} V^{\advspace \authspace'}\ket{y}^{\authspace'})\ket{\sigma}^{\authspace \advspace} \otimes \ket{\Phi}^{\authspace' \authspace''}
	$$
	\item The adversary discards the $\authspace' \authspace''$ register:
	$$
	\frac{1}{N} \sum_{y} (\Id^{\authspace} \otimes \bra{y}^{\authspace'} V^{\advspace \authspace'}\ket{y}^{\authspace'})\ket{\sigma}^{\authspace \advspace}
	$$
\end{enumerate}
This is precisely the state $\Gamma_V \ket{\sigma}^{\authspace \advspace}$, and the adversary described above never touches the $\authspace$ register, so it is ideal.

\subsection{Derandomizing the analysis using approximate unitary designs}

The analysis of this scheme is nearly complete; however, the main missing component is that the analysis above assumes that the authentication scheme uses a truly random unitary $U$ to scramble the message state and the tag. Unfortunately, sampling a truly random unitary on $n$ qubits and applying it is infeasible: only a vanishing fraction of unitaries are succinctly describable or are efficiently computable. 

The authentication scheme instead samples a unitary from a \emph{unitary design}, discussed earlier. These are efficiently sampleable, efficiently computable ensembles of unitaries that are \emph{pseudorandom}: they fool polynomials of low degree. 

It won't be necessary to present formal definitions of a unitary design; we will use them in a black box manner. We will appeal to a general derandomization result of Low who proved that, if one establishes a measure of concentration result for a low degree polynomial $f$ that's evaluated on a Haar-random unitary, then it still satisfies (nearly) the same measure of concentration when $f$ is evaluated on a unitary drawn from an approximate $t$-design. More formally:

\begin{theorem}[\cite{low2009large}]
\label{thm:derandomization}
Let $f: U(N) \to \R$ be a polynomial of degree $K$. Let $f(U) = \sum_i \alpha_i M_i(U)$ where $M_i(U)$ are monomials and let $\alpha(f) = \sum_i |\alpha_i|$. Suppose that $f$ has probability concentration
$$
	\Pr_{U \sim \nu_{Haar}} ( |f - \mu| \geq \delta) \leq C \exp(-a\delta^2)
$$
and let $\mu$ be an $\eps$-approximate unitary $t$-design. Then
$$
	\Pr_{U \sim \mu} ( | f - \mu| \geq \delta) \leq \frac{1}{\delta^{2m}} \left ( C \left( \frac{m}{a} \right)^m + \eps ( \alpha + |\mu|)^{2m} \right)
$$
for integer $m$ with $2mK \leq t$.
\end{theorem}

Furthermore, there exist efficient constructions of approximate $t$-unitary designs, for any $t$. 
\begin{theorem}[\cite{brandao2012local}]
\label{thm:designs}
	For every $\eps$, $t$, and $n$, there exists a finite set of unitaries $D_{\eps,t,n} \subset U(N)$ for $N = 2^n$, and a probability distribution $\mu_{\eps,t,n}$ over $D_{\eps,t,n}$ such that
	\begin{enumerate}
		\item $\mu_{\eps,t,n}$ is an $\eps$-approximate $t$-unitary design.
		\item $\mu_{\eps,t,n}$ can be sampled from in $\poly(n,t,\log 1/\eps)$ time
		\item Each unitary $U \in D_{\eps,t,n}$ can be implemented by a quantum circuit acting on $n$ qubits of size at most $O(n \log(4t)^2 t^9 (2nt + \log(1/\eps)))$.
	\end{enumerate}
\end{theorem}

We combine these two theorems to prove our final result:
\begin{theorem}[Restatement of Theorem~\ref{thm:design}]
	The unitary design scheme is efficiently computable, and is $2^{-s/2}$-totally authenticating.
\end{theorem}
\begin{proof}
Note that $f(U)$ is a polynomial of degree $4$ in the entries of $U$. We compute $\alpha(f)$ by computing $\alpha(f_0)$, $\alpha(g)$, and $\alpha(h)$ where $f_0 = \bra{\rho}\Gamma_V^\dagger \Gamma_V \ket{\rho}$ is a constant, $g(U) = \bra{\rho}\Gamma_V^\dagger \Lambda_U \ket{\rho}$, and $h(U) = \bra{\rho}\Lambda_U^\dagger \Lambda_U \ket{\rho}$. Clearly, $\alpha(f) \leq \alpha(f_0) + 2\alpha(g) + \alpha(h)$. 

Since $f_0$ is a constant function, $\alpha(f_0)$ is at most $|f_0| \leq 1$. We turn to $g$. Let $\{ \ket{x} \}$ be a basis for $\msgspace$. Then for $x, x'$, define the operator $T^{xx'} = \bra{\varphi_x} \Gamma_V^\dagger V \ket{\varphi_{x'}}$ to be the linear operator that maps $\authspace$ to $\authspace$ (recall that $\ket{\rho} = \sum_x \rho_x \ket{x} \otimes \ket{\varphi_x}$). Then,
\begin{align}
	g(U) &= \sum_{x,x'} \rho_x^* \rho_{x'} \bra{\psi_x} T^{xx'} \ket{\psi_{x'}} \\
		&= \sum_{x,x',y,y'} \rho_x^* \rho_{x'} T^{xx'}_{yy'} U_{yx}^* U_{y'x'}
\end{align}
For every $x,x',y,y'$, we have a distinct monomial $U_{yx}^* U_{y'x'}$, and the corresponding coefficient is $\rho_x^* \rho_{x'} T^{xx'}_{yy'}$, which has absolute value at most $1$. Therefore $\alpha(g) \leq M^2 N^2$. 

Now we turn to $h(U)$. Recall that
\begin{align}
	h(U) &= \sum_{x',x''} \rho_{x'}^* \rho_{x''} \sum_x (\bra{\psi_{x'}} \otimes \bra{\varphi_{x'}})  V^\dagger \ket{\psi_{x}} \bra{\psi_{x}} V (\ket{\psi_{x''}} \otimes \ket{\varphi_{x''}}) \\
		&= \sum_{i,j,x',x''} \rho_{x'}^* \rho_{x''} U_{ix'}^* U_{jx''} \sum_x (\bra{i} \otimes \bra{\varphi_{x'}})  V^\dagger \ket{\psi_{x}} \bra{\psi_{x}} V (\ket{j} \otimes \ket{\varphi_{x''}})
\end{align}
where $\ket{\psi_x} = \sum_i U_{ix} \ket{i}$, $\ket{\psi_{x'}} = \sum_i U_{ix'} \ket{i}$ and $\ket{\psi_{x''}} = \sum_j U_{jx''} \ket{j}$. Define $\ket{\tau^{ix'}} = V\ket{i} \otimes \ket{\varphi_{x'}}$ and $\ket{\tau^{jx''}} = V\ket{j} \otimes \ket{\varphi_{x''}}$. Then we have
\begin{align}
	h(U) &= \sum_{i,j,i',j'} \sum_{x,x',x''} U_{ix'}^* U_{jx''} U_{i'x} U_{j'x}^* \rho_{x'}^* \rho_{x''} \ip{\tau^{ix'}}{i'} \ip{j'}{\tau^{jx''}} \\
		&= \sum_{i,j,i',j'} \sum_{x,x',x''} U_{ix'}^* U_{jx''} U_{i'x} U_{j'x}^* \rho_{x'}^* \rho_{x''} \sum_z (\tau^{ix'}_{i'z})^* \,\, \tau^{jx''}_{j'z}
\end{align}
where we alternatively write $\ket{\tau^{ix'}} = \sum_z \tau^{ix'}_{i'z} \ket{i',z}$ and $\ket{\tau^{jx''}} = \sum_z \tau^{jx''}_{j'z} \ket{j',z}$. For every choice of $i,j,i',j',x,x',x''$, we have a distinct monomial, and the associated coefficient has norm at most
$$
\big | \rho_{x'}^* \rho_{x''} \sum_z (\tau^{ix'}_{i'z})^* \,\, \tau^{jx''}_{j'z} \big |^2 \leq \left ( \sum_z |\tau^{ix'}_{i'z}|^2 \right) \cdot \left ( \sum_z |\tau^{jx''}_{j'z}|^2 \right) \leq 1.
$$
Thus $\alpha(h)$ is at most $M^3 N^4$. This implies that $\alpha(f) \leq O(N^7)$. 

Now we are ready to leverage Theorems~\ref{thm:derandomization} and~\ref{thm:designs}. in Lemma~\ref{lem:side_inf_concentration} we proved that function $f(U) = \| \Gamma_V \ket{\rho} - \Gamma_U \ket{\rho} \|_2^2$ has probability concentration
$$
	\Pr_{U \sim \nu_{Haar}} ( |f - \mu| \geq \delta) \leq 4 \exp(-C N\delta^2)
$$
where $C$ is a universal constant. Thus our parameters are:
\begin{enumerate}
	\item (Average of $f$) $\mu = M/N$
	\item (Error in probability concentration) $\delta = \sqrt{M/N}$
	\item (Degree of $f$) $K = 4$
	\item (Probability concentration exponent) $a = CN$
	\item (Norm of $f$) $\alpha(f) = O(N^7)$
\end{enumerate}
We will set $m = 1$, $\eps = N^{-17}$, and $t = 8$. 

By Theorem~\ref{thm:designs}, there exists a distribution $\mu_{\eps,t,n}$ over unitaries acting on $n$ qubits that forms an efficiently computable $\eps$-approximate $t$-unitary design. Then, plugging everything into Theorem~\ref{thm:derandomization}, we have that
\begin{align}
	\Pr_{U \sim \mu_{\eps,t,n}} ( f \geq M/N + \sqrt{M/N}) \leq   O(1/M)
\end{align}
Note that $M/N = 2^{-s}$. 
\end{proof}


\section{A lifting theorem for authentication}
\label{sec:lifting}

\newcommand{\A}{\mathcal{A}}
\newcommand{\B}{\mathcal{B}}
\newcommand{\Clifford}{\mathcal{C}}
\newcommand{\pauli}{\mathsf{Pauli}}

We will prove a \emph{lifting theorem} which shows that a weak form of authentication security that doesn't take into account quantum side information actually implies stronger security against quantum side information. The initial weak form of security is very weak indeed: as long as we have the authentication scheme can securely authenticate a \emph{single} state (namely, one half of the maximally entangled state), in a key-averaged manner, then we can actually obtain an authentication scheme that can authenticate all states --- even those that are entangled with the adversary. 

Specifically, we show that this weak authentication security implies the security definition of~\cite{dupuis2012actively}, which we reproduce here:

\begin{definition}[\cite{dupuis2012actively} security definition]
	An authentication scheme $(\auth,\ver)$ is $\eps$-secure according to the~\cite{dupuis2012actively} definition iff for all initial message states $\ket{\rho}^{\msgspace \safespace \advspace}$, for all adversaries $\superop \in \metalinmap(\authspace \advspace,\authspace \advspace)$, there exists an oblivious adversary $\id$ such that the following are $\eps$-close in trace distance:
\begin{itemize}
	\item (Real experiment) $\Ex_k \left[ \ver_k \circ \superop \circ \auth_k \right](\rho^{\msgspace \safespace \advspace})$
	\item (Ideal experiment) $\id(\rho^{\msgspace \safespace \advspace})$
\end{itemize}
where $\auth_k$ acts on $\msgspace$, $\ver_k$ acts on $\authspace$, and both act as the identity on $\safespace \advspace$.
\end{definition}

The difference between this definition and total authentication is that the key is averaged over in the~\cite{dupuis2012actively} definition.

There is a minor caveat: we do not prove this implication for all authentication schemes. Instead, we prove it for authentication schemes \emph{composed} with a Pauli randomization step. If $(\auth,\ver)$ is an authentication scheme, we call this composed scheme $\pauli + (\auth,\ver)$, and it behaves as follows:

The secret key for $\pauli + (\auth,\ver)$ consists of the key $k$ for $(\auth,\ver)$, as well as a new, independent key $k'$. The procedure to authenticate a message register $\msgspace$ behaves as follows: first, the key $k'$ is used to choose a random unitary from the Pauli group that acts on the space $\msgspace$.\footnote{For simplicitly let us think of $\msgspace$ as $(\C^2)^{\otimes n}$ (i.e., $n$ qubits). Then the Pauli group consists of all operators of the form $X^p Z^q$, where $p,q \in \{0,1\}^n$. Here, the operator $X^p$ is defined to be the tensor product of $X_j^{p_j}$, where $X_j$ is the $X$ Pauli operator acting on the $j$'th qubit. $Z^q$ is defined similarly.} We call this the \emph{Pauli randomization step}. Next, the key $k$ is used to apply $\auth_k$ to the register $\msgspace$ to produce a state in the $\authspace$ register. This is the authenticated state, which is then subject to attack by the adversary.

To un-authenticate, the $\ver_k$ procedure is applied. Note that this is not a unitary operation, but includes the projection on the receiver's acceptance (see the Preliminaries for a discussion of this). Finally, the Pauli randomization is undone using the key $k'$.

\begin{theorem}[Lifting weak authentication to total authentication]
\label{thm:lifting2}
	Let $(\auth,\ver)$ be an authentication scheme, and suppose the composed scheme $\pauli + (\auth,\ver)$ satisfies the following security guarantee:
 for all adversaries $\superop \in \metalinmap(\authspace \advspace,\authspace \advspace)$, for all adversary ancilla qubits $\ket{\theta}^{\advspace \advspace'}$, there exists an oblivious adversary $\id$ acting on $\advspace$ only such that the following are $\eps$-close in trace distance:
\begin{itemize}
	\item (Real experiment) $\Ex_{k,k'} \left[ \pauli_{k'}^\dagger \circ \ver_k \circ \superop \circ \auth_k \circ \pauli_{k'} \right](\ketbra{\Phi}{\Phi}^{\msgspace \B} \otimes \ketbra{\theta}{\theta}^{\advspace \advspace'})$
	\item (Ideal experiment) $\ketbra{\Phi}{\Phi}^{\msgspace \B} \otimes \id(\ketbra{\theta}{\theta}^{\advspace \advspace'})$
\end{itemize}
where $\B$ is a Hilbert space isomorphic to $\msgspace$, and $\ket{\Phi}^{\msgspace \B}$ is the maximally entangled state. 

Then, the composed scheme $\pauli + (\auth,\ver)$ is a $\eps$-secure according to the~\cite{dupuis2012actively} definition.
\end{theorem}
\begin{proof}

We wish to argue that $\pauli + (\auth,\ver)$ is secure according to the~\cite{dupuis2012actively} definition. To that end, let $\ket{\rho}^{\msgspace \safespace \advspace}$ be a state that is entangled with the adversary, and let $\superop \in \metalinmap(\authspace \advspace,\authspace \advspace)$ be an arbitrary adversary.

\medskip
\noindent \textbf{Real experiment.} Consider an execution of the $\pauli + (\auth,\ver)$ protocol on input $\rho$ with adversary $\superop$. We can represent this execution as a channel that takes in three registers ($\advspace$, $\msgspace$, and $\safespace$), and outputs the same three registers. Recall that the $\ver$ operation is not a unitary, because we project on the protocol accepting (see the Preliminaries for a discussion of this). This channel is represented diagrammatically in Figure~\ref{fig:real}, where time advances left to right. We call this execution the Real Experiment. Our goal is to conclude that the output of the Real Experiment is approximately the result of an ideal adversary. 

The final state of the Real Experiment is
\begin{equation}
\label{eq:real_final}
\Ex_{k,k'} (\pauli_{k'}^\dagger \circ \ver_k \circ \superop \circ \auth_k \circ \pauli_{k'}) \left(\rho^{\advspace \msgspace \safespace}\right).
\end{equation}

\begin{figure}[H]
\centering
\includegraphics[width=17cm]{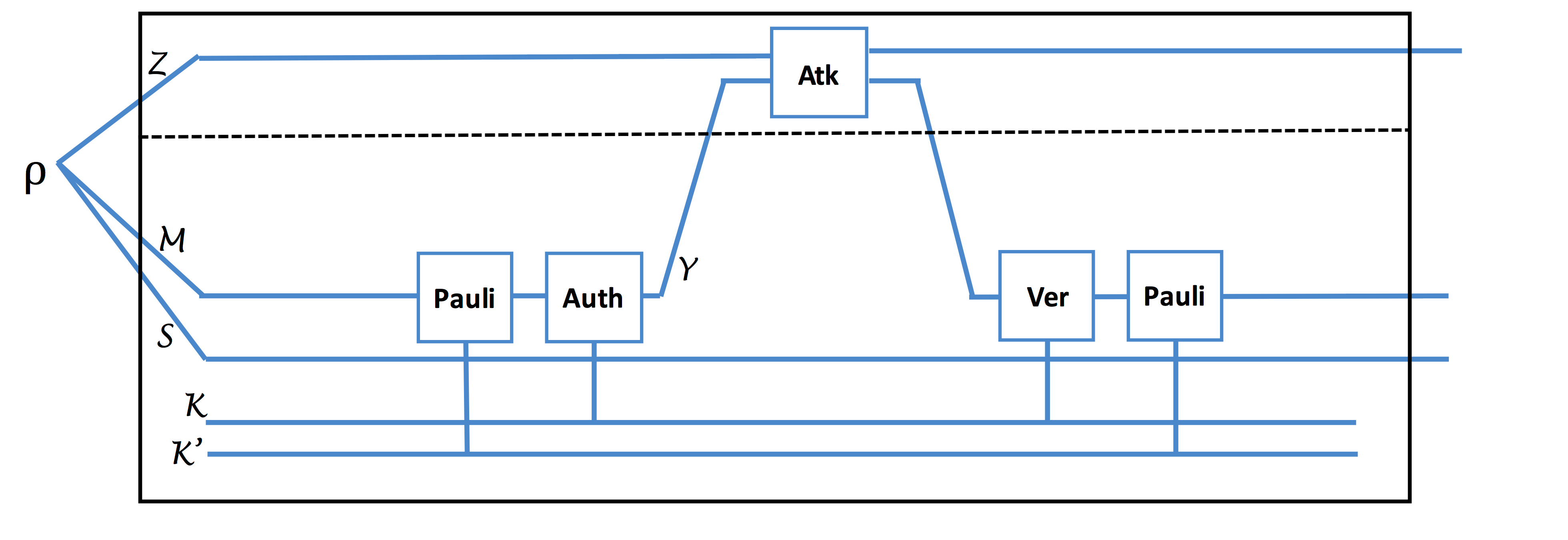}
\caption{Real Experiment}
\label{fig:real}
\end{figure}

\medskip
\noindent \textbf{Hybrid 1.} Next we argue that the intput and output behavior of the real protocol is indistinguishable from the next protocol, demonstrated in Figure~\ref{fig:hybrid1}. Here, inside the protocol, instead of first authenticating the $\msgspace$ part of $\rho$, the protocol introduces an ancillary maximally entangled state $\ket{\Phi}^{\A \B}$ two new registers $\A$, $\B$, where both $\A$ and $\B$ are isomorphic to $\msgspace$. The protocol then \emph{teleports} $\rho^{\msgspace}$ from register $\msgspace$ to register $\A$ by making a Bell measurement on the $\msgspace$ and $\B$ registers.

\begin{figure}[H]
\centering
\includegraphics[width=17cm]{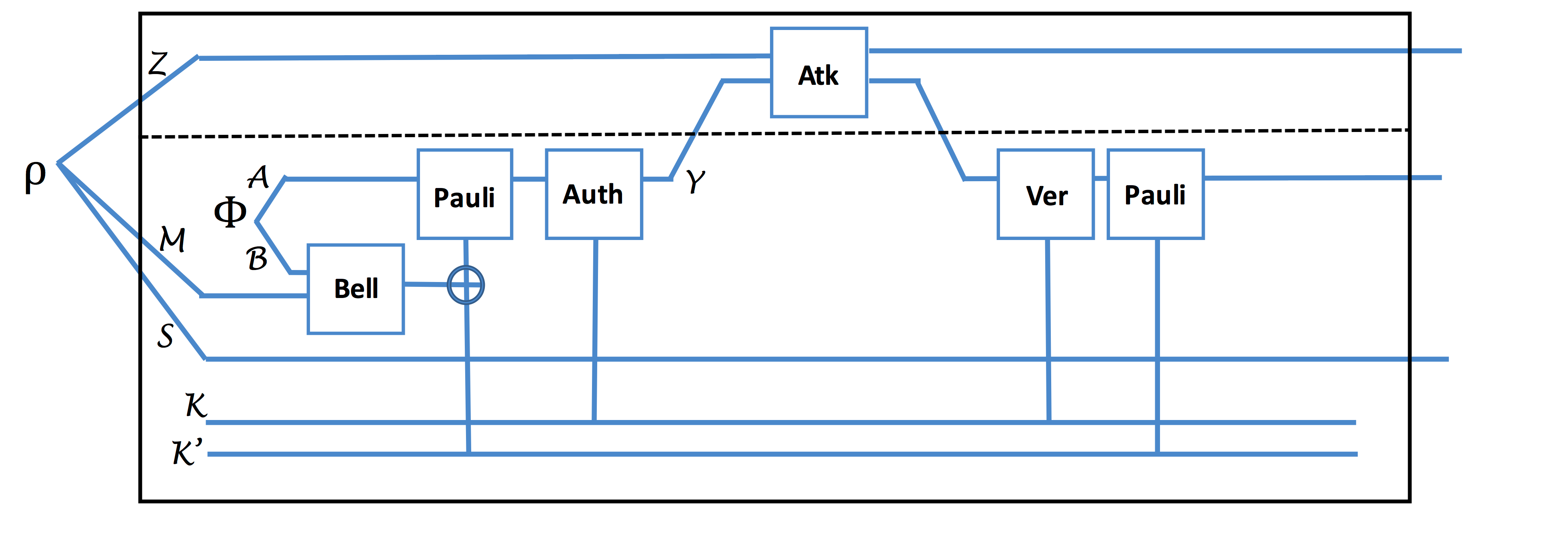}
\caption{Hybrid 1}
\label{fig:hybrid1}
\end{figure}

The teleportation sub-circuit is implemented by first making a Bell measurement on the $\B \msgspace$ registers. After the measurement, the $\A$ register now contains the density matrix $\rho^\msgspace$, conjugated by a Pauli error $X^p Z^q$ that depend on the outcomes $(p,q)$ of the Bell measurement. We can write the state of the $\advspace \A \safespace$ registers as follows:
$$
	\Ex_{k,k',p,q} \pauli_{p,q}^\A (\rho^{\advspace \A \safespace})
$$
where $\pauli_{p,q}(\rho) = X^p Z^q (\rho) (X^p Z^q)^\dagger$. Note that the outcomes $p$ and $q$ are independent of $k$ and $k'$. 

Subsequently, the outcomes $(p,q)$ are used to apply the correction $\pauli_{p,q}^\dagger$.\footnote{In Figure~\ref{fig:hybrid1}, we combine the Pauli correction with the Pauli operation that is part of the authentication scheme into one operator. Note that this merged Pauli operator is $\pauli_{k'} \pauli_{p,q}^\dagger = -\pauli_{k' \oplus (p,q)}$. The overall phase is immaterial.} After this Pauli correction, the effect of the teleportation process is to have swapped the registers $\msgspace$ and $\A$. The residual entangled state on registers $\B \msgspace$ is now some Bell basis state, which we ignore for the rest of the protocol. The state $\rho^{\advspace \A \safespace}$ undergoes the Real Experiment protocol, so therefore the final state of Hybrid 1 is going to be identical to the output state of the Real Experiment in~\eqref{eq:real_final} (up to relabeling of registers):
\begin{equation}
\label{eq:hybrid1_final}
\Ex_{k,k',p,q} (\pauli_{k'}^\dagger \circ \ver_k \circ \superop \circ \auth_k \circ \pauli_{k'}) \left(\rho^{\advspace \A \safespace}\right).
\end{equation}
Although the measurement outcomes $p$ and $q$ are not used inside the expectation, we still keep track of them (as will be useful in the next Hybrid).

\medskip
\noindent \textbf{Hybrid 2.} We can rewrite the final state of Hybrid 1 in the following way:
\begin{align*}
\Ex_{k,k',p,q} (\pauli_{k' \oplus (p,q)}^\dagger \circ \ver_k \circ \superop \circ \auth_k \circ \pauli_{k' \oplus (p,q)}) \left(\rho^{\advspace \A \safespace}\right).
\end{align*}
This is because the distribution $k'$ and $k' \oplus (p,q)$ are identical, and these random strings are independent of $k$. But observe that this final state corresponds to the protocol diagrammed in Figure~\ref{fig:hybrid2}, where the Pauli corrections are not done immediately after the Bell measurement, but rather deferred to \emph{after} the $\pauli_{k'}$ operation.
%
%
%

\begin{figure}[H]
\centering
\includegraphics[width=17cm]{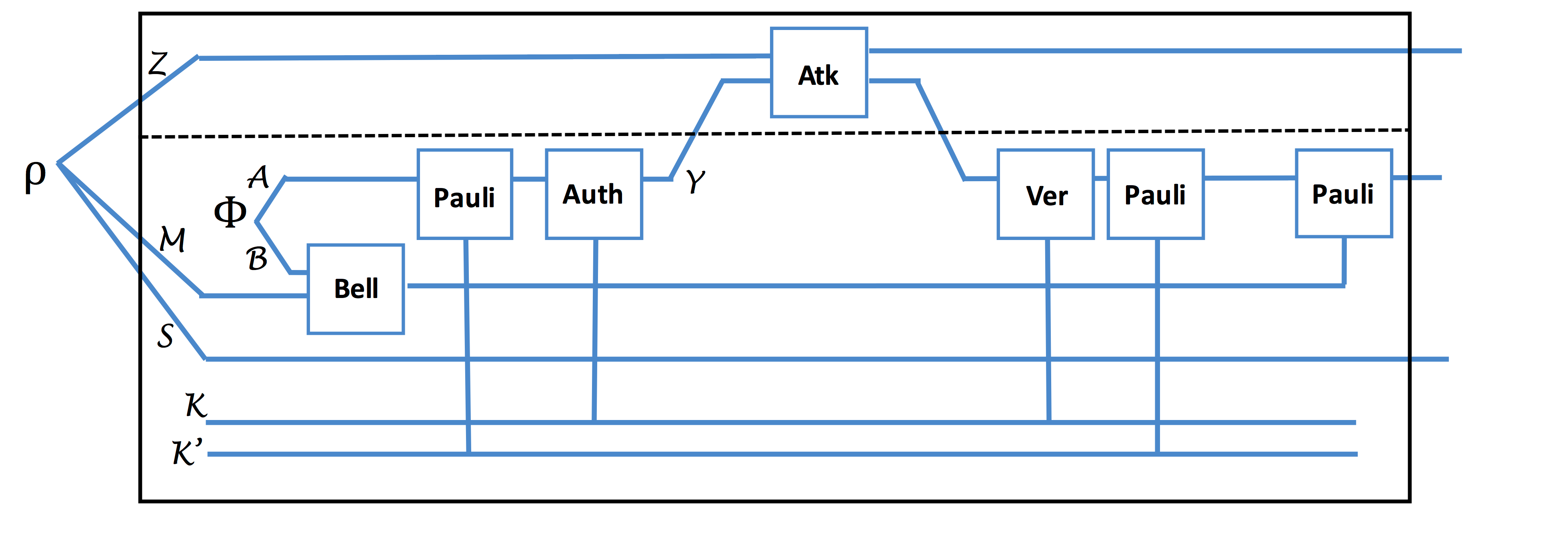}
\caption{Hybrid 2}
\label{fig:hybrid2}
\end{figure}

\medskip
\noindent \textbf{Hybrid 3.} But since the Pauli corrections happen after the verification procedure, the Bell measurement commutes with $\auth_k$, the attack $\superop$,the verification procedure $\ver_k$, and $\pauli_{k'}$. Therefore we can push the Bell measurement to just right before the Pauli corrections, as in Figure~\ref{fig:hybrid3}. But now observe that the state of the protocol, right up to the vertical dashed line, can be written as follows:
$$
	\Ex_{k,k'} (\pauli_{k'}^\dagger \circ \ver_k \circ \superop \circ \auth_k \circ \pauli_{k'}) (\ketbra{\Phi}{\Phi}^{\A \B} \otimes \rho^{\advspace \msgspace \safespace})
$$
where the sequence of operations $\pauli_{k'}^\dagger \circ \ver_k \circ \superop \circ \auth_k \circ \pauli_{k'}$ only interact with the $\A$ and $\advspace$ registers. 

\begin{figure}[H]
\centering
\includegraphics[width=17cm]{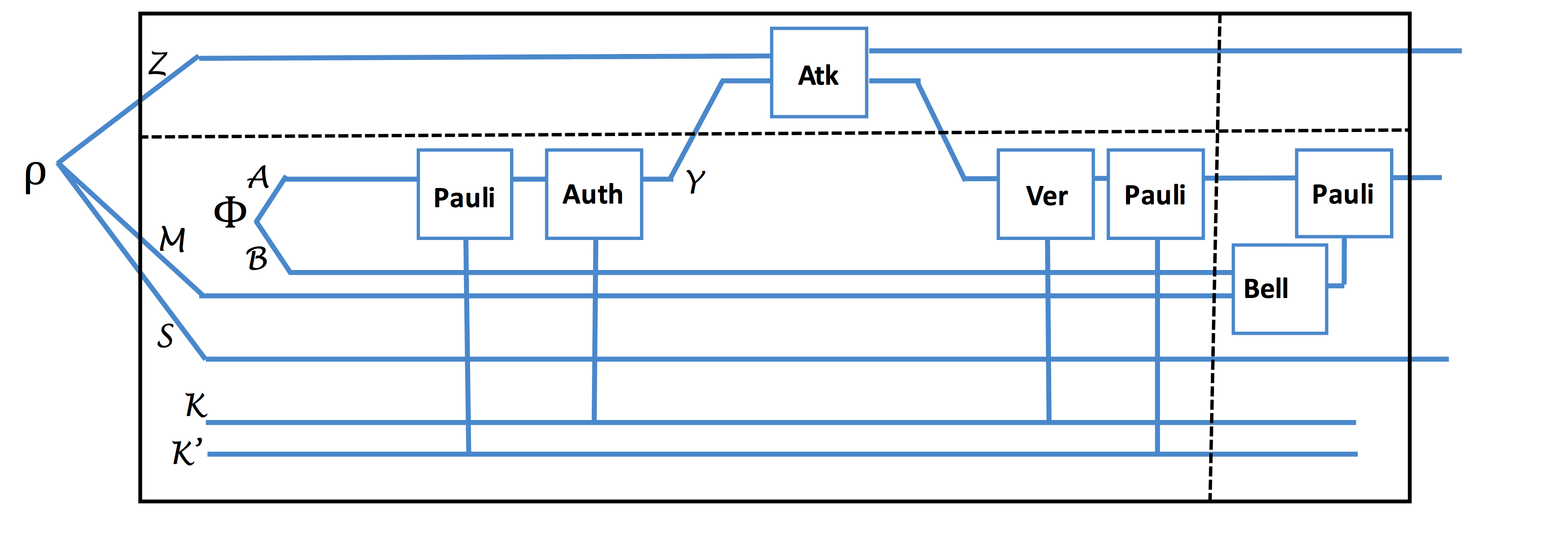}
\caption{Hybrid 3}
\label{fig:hybrid3}
\end{figure}

The state up to that point looks like the sender authenticated the $\A$ half of maximally entangled state $\ket{\Phi}^{\A \B}$ using the secret keys $k,k'$, and kept the other half. The adversary then applies an attack $\superop$ on the authentication of $\Phi^A$, as well the $\advspace$ part of $\rho$ (which is unentangled with $\Phi$). After the adversary's attack, the message register undergoes the verification procedure.

But this is precisely the situation we are assuming we have a security guarantee about. Thus there exists an ideal oblivious adversary $\id$ acting on $\advspace$ only such that the state of the protocol (averaged over $k, k'$), up to the vertical dashed line, is $\eps$-close to
$$
	\ketbra{\Phi}{\Phi}^{\A \B} \otimes \id(\rho^{\advspace \msgspace \safespace}).
$$
The teleportation circuit is then applied to this state. The final state of Hybrid 4 is therefore $\eps$-close to 
$$
	\id(\rho^{\advspace \A \safespace})
$$
where we have traced out the $\B \msgspace$ registers. However this implies that the final state of the Real Experiment must be $\eps$-close to $\id(\rho^{\advspace \A \safespace})$. 

We have thus proved that $\pauli + (\auth,\ver)$ is a $\eps$-secure according to the~\cite{dupuis2012actively} definition of security. This finishes the proof.

\end{proof}

We remark that this proof strategy is heavily inspired by the reduction given in~\cite{hayden2011universal}.

\section{Open problems}

We close with some open problems:
\begin{enumerate}
	\item Is three-wise independence necessary for the Wegman-Carter scheme to be quantumly secure?
	
	\item We showed that the Auth-QFT-Auth scheme achieves total authentication (with outer key leakage) when the inner authentication scheme is instantiated with the Wegman-Carter scheme using threewise-independent hashing. Can one show that Auth-QFT-Auth achieves total authentication when both inner and outer authentication schemes are \emph{arbitrary} authentication schemes secure relative to the computational basis?
	
	\item We showed that the scheme based on unitary 8-design achieves total authentication. Can one show the same for unitary 2-designs? Does the Clifford scheme achieve total authentication?
	
	\item Under what circumstances can the key be reused in any of the protocols presented in this paper, when the receiver rejects the state? For example, we conjecture that in the unitary design protocol, much of the key can be reused.
	
	\item Our security definitions are specific to ``one-time'' authentication schemes (although the key reuse properties allow multiple uses). Are there natural ``many-time'' versions of our security definitions?
	
	\item Does total authentication satisfy \emph{Universally Composable} security (as defined in~\cite{ben2005universal,unruh2010universally})?
	
\end{enumerate}

\paragraph{Acknowledgments.} We thank Debbie Leung for kindly sharing a manuscript of~\cite{hayden2011universal}, and thank Anne Broadbent, Debbie, Patrick Hayden, Fang Song, and Thomas Vidick for useful feedback.

\bibliographystyle{alpha}
\bibliography{qmacs}

\newcommand{\etalchar}[1]{$^{#1}$}
\begin{thebibliography}{KLLNP16}

\bibitem[ABE10]{aharonov2008interactive}
Dorit Aharonov, Michael {Ben-Or}, and Elad Eban.
\newblock Interactive proofs for quantum computations.
\newblock In {\em Proceedings of Innovations in Computer Science}. Tsinghua
  University Press, 2010.

\bibitem[ABW09]{ambainis2009nonmalleable}
Andris Ambainis, Jan Bouda, and Andreas Winter.
\newblock Nonmalleable encryption of quantum information.
\newblock {\em Journal of Mathematical Physics}, 50(4):042106, 2009.

\bibitem[BCG{\etalchar{+}}02]{barnum2002authentication}
Howard Barnum, Claude Cr{\'e}peau, Daniel Gottesman, Adam Smith, and Alain
  Tapp.
\newblock Authentication of quantum messages.
\newblock In {\em The Proceedings of the 43rd Annual IEEE Foundations of
  Computer Science, 2002.}, pages 449--458. IEEE, 2002.

\bibitem[BCG{\etalchar{+}}06]{ben2006secure}
Michael {Ben-Or}, Claude Cr{\'e}peau, Daniel Gottesman, Avinatan Hassidim, and
  Adam Smith.
\newblock Secure multiparty quantum computation with (only) a strict honest
  majority.
\newblock In {\em 2006 47th Annual IEEE Symposium on Foundations of Computer
  Science (FOCS'06)}, pages 249--260. IEEE, 2006.

\bibitem[BDF{\etalchar{+}}11]{Boneh2011qrom}
Dan Boneh, \"{O}zg\"{u}r Dagdelen, Marc Fischlin, Anja Lehmann, Christian
  Schaffner, and Mark Zhandry.
\newblock {Random Oracles in a Quantum World}.
\newblock In {\em Proceedings of ASIACRYPT}, 2011.

\bibitem[Bee97]{beenakker1997random}
Carlo~WJ Beenakker.
\newblock Random-matrix theory of quantum transport.
\newblock {\em Reviews of modern physics}, 69(3):731, 1997.

\bibitem[BGS13]{broadbent2013quantum}
Anne Broadbent, Gus Gutoski, and Douglas Stebila.
\newblock Quantum one-time programs.
\newblock In {\em Advances in Cryptology--CRYPTO 2013}, pages 344--360.
  Springer, 2013.

\bibitem[BHH12]{brandao2012local}
Fernando~GSL Brandao, Aram~W Harrow, and Michal Horodecki.
\newblock Local random quantum circuits are approximate polynomial-designs.
\newblock {\em arXiv preprint arXiv:1208.0692}, 2012.

\bibitem[BHL{\etalchar{+}}05]{ben2005universal}
Michael {Ben-Or}, Micha{\l} Horodecki, Debbie~W Leung, Dominic Mayers, and
  Jonathan Oppenheim.
\newblock The universal composable security of quantum key distribution.
\newblock In {\em Theory of Cryptography Conference}, pages 386--406. Springer,
  2005.

\bibitem[BJ15]{broadbent2015quantum}
Anne Broadbent and Stacey Jeffery.
\newblock Quantum homomorphic encryption for circuits of low t-gate complexity.
\newblock In {\em Annual Cryptology Conference}, pages 609--629. Springer,
  2015.

\bibitem[BW16]{broadbent2016efficient}
Anne Broadbent and Evelyn Wainewright.
\newblock Efficient simulation for quantum message authentication.
\newblock {\em arXiv preprint arXiv:1607.03075}, 2016.

\bibitem[BZ13a]{boneh2013quantum}
Dan Boneh and Mark Zhandry.
\newblock Quantum-secure message authentication codes.
\newblock In {\em Advances in Cryptology--EUROCRYPT 2013}, pages 592--608.
  Springer, 2013.

\bibitem[BZ13b]{boneh2013secure}
Dan Boneh and Mark Zhandry.
\newblock Secure signatures and chosen ciphertext security in a quantum
  computing world.
\newblock In {\em Advances in Cryptology--CRYPTO 2013}, pages 361--379.
  Springer, 2013.

\bibitem[DFNS13]{damgaard2013superposition}
Ivan Damg{\aa}rd, Jakob Funder, Jesper~Buus Nielsen, and Louis Salvail.
\newblock Superposition attacks on cryptographic protocols.
\newblock In {\em Information Theoretic Security}, pages 142--161. Springer,
  2013.

\bibitem[DNS12]{dupuis2012actively}
Fr{\'e}d{\'e}ric Dupuis, Jesper~Buus Nielsen, and Louis Salvail.
\newblock Actively secure two-party evaluation of any quantum operation.
\newblock In {\em Advances in Cryptology--CRYPTO 2012}, pages 794--811.
  Springer, 2012.

\bibitem[GHS15]{gagliardoni2015semantic}
Tommaso Gagliardoni, Andreas H{\"u}lsing, and Christian Schaffner.
\newblock Semantic security and indistinguishability in the quantum world.
\newblock {\em arXiv preprint arXiv:1504.05255}, 2015.

\bibitem[HLM11]{hayden2011universal}
Patrick~M. Hayden, Debbie~W. Leung, and Dominic Mayers.
\newblock The universal composable security of quantum message authentication
  with key recycling.
\newblock {\em In preparation}, 2011.

\bibitem[KLLNP16]{kaplan2016breaking}
Marc Kaplan, Ga{\"e}tan Leurent, Anthony Leverrier, and Mar{\'\i}a
  Naya-Plasencia.
\newblock Breaking symmetric cryptosystems using quantum period finding.
\newblock {\em arXiv preprint arXiv:1602.05973}, 2016.

\bibitem[Low09]{low2009large}
Richard~A Low.
\newblock Large deviation bounds for k-designs.
\newblock In {\em Proceedings of the Royal Society of London A: Mathematical,
  Physical and Engineering Sciences}, volume 465, pages 3289--3308. The Royal
  Society, 2009.

\bibitem[MS09]{milman2009asymptotic}
Vitali~D Milman and Gideon Schechtman.
\newblock {\em Asymptotic Theory of Finite Dimensional Normed Spaces:
  Isoperimetric Inequalities in Riemannian Manifolds}, volume 1200.
\newblock Springer, 2009.

\bibitem[Unr10]{unruh2010universally}
Dominique Unruh.
\newblock Universally composable quantum multi-party computation.
\newblock In {\em Annual International Conference on the Theory and
  Applications of Cryptographic Techniques}, pages 486--505. Springer, 2010.

\bibitem[WC81]{wegman1981new}
Mark~N Wegman and J~Lawrence Carter.
\newblock New hash functions and their use in authentication and set equality.
\newblock {\em Journal of computer and system sciences}, 22(3):265--279, 1981.

\bibitem[Web15]{webb2015clifford}
Zak Webb.
\newblock The clifford group forms a unitary 3-design.
\newblock {\em arXiv preprint arXiv:1510.02769}, 2015.

\bibitem[Zha12]{zhandry2012qprfs}
Mark Zhandry.
\newblock {How to Construct Quantum Random Functions}.
\newblock In {\em Proceedings of the 53rd IEEE Symposium on Foundations of
  Computer Science (FOCS)}, 2012.

\bibitem[Zhu15]{zhu2015multiqubit}
Huangjun Zhu.
\newblock Multiqubit clifford groups are unitary 3-designs.
\newblock {\em arXiv preprint arXiv:1510.02619}, 2015.

\end{thebibliography}

\end{document}